\RequirePackage{fix-cm}
\documentclass[smallextended]{svjour3}       % onecolumn (second format)
\smartqed % flush right qed marks, e.g. at end of proof
\usepackage{graphicx}
\usepackage[shortcuts]{extdash}

\usepackage{amsmath,amsfonts,amssymb,amstext,bbm}

\usepackage{natbib}
\usepackage{xcolor}

\newtheorem{observation}{Observation}

\usepackage{float}
\floatstyle{ruled}
\newfloat{algorithm}{tbp}{loa}
\floatname{algorithm}{Algorithm}
\usepackage{algorithmic}

\usepackage{multirow}
\usepackage{balance}

\ifdefined\REVISION
\journalname{Autonomous Agents and Multi-Agent Systems}
\newcommand{\rev}[1]{{\color{blue} #1}}
\else
\newcommand{\rev}[1]{#1}
\fi

\newcommand{\related}[1]{\emph{#1}}
\newcommand{\range}[2]{\in\{#1,\dots,#2\}}

\title{Redividing the Cake
\footnote{
A preliminary version of this paper was presented in the 27th International Joint Conference on Artificial Intelligence, IJCAI \citep{segalhalevi2018redividing}.
The following are new in the present paper.
(a) Theorem \ref{thm:rectilinear+}, which generalizes Theorem \ref{thm:rectangle+} from a rectangle to any rectilinear polygon. 
(b) Section \ref{sec:pof-bounds}, which applies the algorithms for cake redivision to obtain upper bounds on the price of fairness in cake-cutting with geometric constraints.
(c) The proof of Lemma \ref{lem:convexity} now uses a recently-introduced algorithm by \citet{cseh2020complexity} to obtain an algorithm with run-time polynomial in the binary representation of the input.
(d) Theorem \ref{thm:connected+} now uses an improved algorithm, which attains $1/2$ proportionality instead of $1/3$-proportionality.
Similarly, the constant in Theorem \ref{thm:rectangle+} is $1/3$ instead of $1/4$,
and the constant in Theorem \ref{thm:convex+} is $1/4$ instead of $1/5$.
}
}

\author{
Erel Segal-Halevi
}

\institute{
Ariel University, 
Ariel 40700, Israel.
\email{erelsgl@gmail.com}
}

\def\FULLVERSION{}

\usepackage{framed}

\newtheorem{remark*}{Remark}
\newtheorem{lemma*}{Lemma}

\makeatletter
\newcommand{\text@hyphens}{\mathcode`\-=`\-\relax}
\newcommand{\id}[1]{\ensuremath{\mathit{\text@hyphens#1}}}
\makeatother

\usepackage{multirow}

\usepackage{graphicx}

\usepackage[overload]{textcase}
\usepackage{afterpage}
\usepackage{color}

\usepackage{float}
\floatstyle{ruled}
\newfloat{algorithm}{tbp}{loa}
\floatname{algorithm}{Algorithm}

\usepackage[shortcuts]{extdash}

\usepackage{framed}

\begin{document}

\maketitle

\begin{abstract}
The paper considers fair allocation of resources that are already allocated in an unfair way. This setting requires a careful balance between the fairness considerations and the rights of the present owners.

The paper presents re-division algorithms that attain various trade-off points between fairness and ownership rights, in various settings differing in the geometric constraints on the allotments: (a) no geometric constraints; (b) connectivity---the cake is a one-dimensional interval and each piece must be a contiguous interval; (c) rectangularity---the cake is a two-dimensional rectangle or rectilinear polygon and the pieces should be rectangles; (d) convexity---the cake is a two-dimensional convex polygon and the pieces should be convex.

These re-division algorithms have implications on another problem: the price-of-fairness---the loss of social welfare caused by fairness requirements. Each algorithm implies an upper bound on the price-of-fairness with the respective geometric constraints.
\end{abstract}

\keywords{Cake-cutting \and Land reform \and Dynamic fair division \and Computational Geometry \and Two-dimensional resource allocation}

\section{Introduction}
Most theoretical works on fair resource allocation consider a one-shot division: the resource is divided once and for all, like a cake that is divided and eaten soon after it comes out of the oven. But in practice, it is often required to re-divide an already-divided resource
(see subsection \ref{sub:dynamic}). One example is a cloud-computing environment, where new agents come and require resources held by other agents. A second example is fair allocation of radio spectrum among several broadcasting agencies: it may be required to re-divide the frequencies to accommodate new broadcasters. A third example is land-reform: large land-estates are held by a small number of landlords, and the government may want to re-divide them to landless citizens. 

In the classic one-shot division setting, there are $n$ agents with equal rights, and the goal is to give each agent a fair share of the cake. A common definition of a ``fair share'' is a piece worth at least $1/n$ of the total cake value, according to the agent's personal valuation function. This fairness requirement is usually termed \emph{proportionality}. When proportionality cannot be attained, it is often (see subsection \ref{sub:partial}) relaxed to \emph{$r$-proportionality}, where $r\in(0,1)$ is a constant independent of $n$, which means that each agent receives at least a fraction $r/n$ of the total.

In contrast, in the re-division setting, there is an existing allocation of the cake among the $n$ agents. This allocation is not necessarily fair; in particular, there may be some agents who do not have any cake. When the cake is re-divided, it may be required to give extra rights to current holders. In particular, it may be required to give each agent the opportunity to keep a substantial fraction of their current value. This may be due either to efficiency reasons (in the cloud computing scenario) or economic reasons (in the radio spectrum scenario) or political reasons (in the land-reform scenario). 
This requirement will be called \emph{ownership}. Given a constant $w\in(0,1)$, 
\emph{$w$-ownership} means that each agent receives at least $w$ times their old value.
What levels of proportionality and ownership can be attained simultaneously? 

\subsection{\textbf{Results: Redivision}}
The first two results (in Section \ref{sec:general}) provide a tight answer to this question.

\begin{proposition}
\label{thm:general-}
For every constants $r,w\in[0,1]$ where $r+w > 1$, it may be impossible to simultaneously guarantee $r$-proportionality and $w$-ownership.
\end{proposition}

\begin{theorem}
\label{thm:general+}
For every 
constants 
$r,w\in[0,1]$ where $r+w\leq 1$, and for every existing allocation of the cake, there exists a division that simultaneously satisfies $r$-proportionality and $w$-ownership. 
Moreover, 
when $r,w$ are  rational numbers, such a division 
can be found using $O(n^2\operatorname{len}(r))$ queries, where $\operatorname{len}(r)$ denotes binary representation length.
\end{theorem}
As an example, taking $r=w=1/2$, it is possible to re-divide the cake, giving each agent at least half their previous value, while simultaneously giving each agent at least $1/(2n)$ of the total cake value. 

The parameters $r,w$ represent the level of balance between two principles: large $r$ means more emphasis on fairness while large $w$ means more emphasis on ownership rights.
The above theorems imply that the re-dividers (e.g. the government)
may choose any level of
fairness and ownership-rights that fit their ideological, political or economic goals, as long as the sum of these fractions is at most 1.

The balance parameters can also be given probabilistic interpretation. Suppose the government wants to do a land reform and needs the agreement of the current landowners. Naturally, the current landowners do not want to give away their lands. However, they may fear that, without land-reform, the landless citizens might revolt and they might lose all their lands. If the landowners believe that the probability of a successful revolt is $1-w$, then they may agree to a land-reform that guarantees $w$-ownership. Theorem \ref{thm:general+} implies that, in this case, it is possible to carry out a land-reform that guarantees $(1-w)$-proportionality.

While Theorem \ref{thm:general+} is encouraging, it ignores an important aspect of practical division problems: geometry. The division it guarantees may be highly fractioned, giving each agent a large number of disconnected pieces. In many practical division problems, e.g. when the resource to divide is time, the agents may need to receive a single connected piece rather than a large number of disconnected ones. Can partial-proportionality and partial-ownership be attained simultaneously with a connectivity constraint? The following proposition (proved in Section \ref{sec:connected}) answers this question negatively.
\begin{proposition}
\label{thm:connected-}
When the cake is a 1-dimensional interval and each piece must be an interval, for every positive constants $r,w\in(0,1)$, it may be impossible to simultaneously satisfy $r$-proportionality and $w$-ownership. 

Moreover, for every $r>0$ and every integer $d\in[n]$, there might be $d$ agents who, in any $r$-proportional division, receive at most a fraction $1/\lfloor{\frac{n}{d}}\rfloor$ of their old value.
\end{proposition}
The latter part of the proposition involves a fairness property much weaker than proportionality, that can be termed 
\emph{positivity}---guaranteeing each agent a piece with a positive value. 
With the connectivity constraint, even this weak fairness requirement is incompatible with $w$-ownership for every constant $w>0$: a positive division might require to give one agent at most $1/n$ of their previous value, give two agents at most $2/n$ of their previous value, give $n/3$ agents at most $1/3$ of their previous value, etc.

Proposition \ref{thm:connected-} motivates the following weaker ownership requirement: for every $d$, at least $n-d$ agents receive at least a fraction $1/\lfloor{n\over d}\rfloor$ of their old value. For example (taking $d=n/3$ and assuming all quotients are integers), at least $2n/3$ agents should receive at least $1/3$ of their old value.
This criterion is inspired by the ``90th percentile'' criterion common in Service-Level-Agreements and Quality-of-Service analysis, e.g. \citep{Zhang2014SMiTe,Delimitrou2014Quasar}. It can also be justified by political reasoning: in a democratic country, it may be sufficient to win the support of a sufficiently large majority.

The following results almost match this relaxed ownership criterion. Formally, \rev{let us define the} \textbf{democratic ownership} property as follows: for every integer $d\range{1}{n-1}$, at least  $n-d$ agents receive \rev{more than} a fraction $1/\lceil{n\over d}\rceil$ of their previous value.
\rev{Democratic-ownership corresponds to the best guarantee one could hope for given} Proposition \ref{thm:connected-}; the only difference is that 
in the upper bound the fraction is rounded down  ($1/\lfloor{n\over d}\rfloor$) while in democratic-ownership the fraction is rounded up.

\begin{theorem}
\label{thm:connected+}
When the cake is a 1-dimensional interval 
and each piece must be an interval, for every existing allocation of the cake, 
it is possible to find in time $O(n^2 \log{n})$ 
a division simultaneously satisfying democratic-ownership and $1/2$-proportionality.
\end{theorem}
It is an open question whether democratic-ownership is compatible with $r$-proportionality for some constant $r>1/2$.

Theorem \ref{thm:connected+}, like most works in cake-cutting, assumes that the cake is 1-dimensional. In realistic division scenarios, the cake is often 2-dimensional and the pieces should have a pre-specified geometric shape, such as a rectangle or a convex polygon. Rectangularity and convexity requirements are sensible when dividing land, exhibition space in museums, advertisement space in newspapers and even virtual space in web-pages. Moreover, in the frequency-range allocation problem, it is possible to allocate frequency ranges for a limited time-period; the frequency-time space is two-dimensional and it makes sense to require that the ``pieces'' are rectangles in this space \citep{Iyer2009Procedure}. 

2-dimensional cake-cutting introduces new challenges over the traditional 1-dimensional setting. As an example, in one dimension, it can be assumed that the initial allocation is a partition of the entire cake; this is without loss of generality, since any ``blank'' (unallocated part) can be attached to a neighboring allocated interval without harming its shape or value. However, in two dimensions, the initial allocation might contain blanks that cannot be attached to any allocated piece due to the rectangularity or convexity constraints. For example, 
suppose the cake is the large rectangle in Figure \ref{fig:disposal}.
\begin{figure}
\begin{center}
\includegraphics[width=.45\columnwidth]{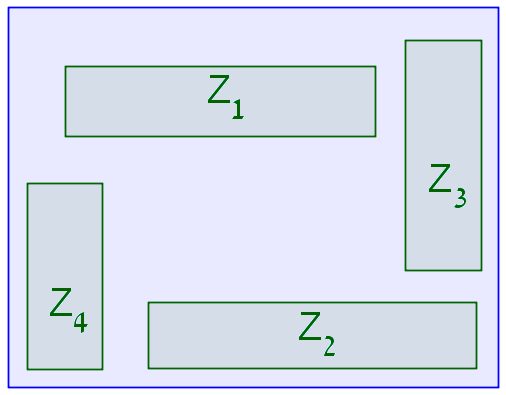}
\end{center}
\caption{
\label{fig:disposal}
With geometric constraints, an efficient allocation might leave some cake unallocated.
~~~
{\footnotesize All figures were made with GeoGebra 5 \citep{GeoGebra5}.}
}
\end{figure}
There are 4 agents and each agent $i$ has positive value-density only inside the rectangle $Z_i$. The most reasonable division (e.g. the only Pareto-efficient division) is to give each $Z_i$ entirely to agent $i$. But, this allocation leaves a blank in the center of the cake, and this blank cannot be attached to any allocated piece due to the rectangularity constraint. 
This counter-intuitive scenario cannot happen in a one-dimensional cake. Handling such cases requires new geometry-based tools. With such tools, the redivision problem can be solved in two common 2-dimensional settings (Section \ref{sec:rectangle}):
\begin{theorem}
\label{thm:rectangle+}
When the cake is a rectangle 
and each piece must be a parallel rectangle,
for every existing allocation of the cake, 
it is possible to find in time $O(n^2 \log{n})$ 
a division simultaneously satisfying democratic-ownership and $1/3$-proportionality.
\end{theorem}
\begin{theorem}
\label{thm:convex+}
When the cake is a 2-dimensional convex polygon 
and each piece must be convex,
for every existing allocation of the cake, 
there exists a division simultaneously satisfying democratic-ownership and $1/4$-proportionality.
\end{theorem}
\begin{remark}
In the interval, rectangle and convex settings, the geometric constraints are mostly harmless without the ownership requirement: when the cake is an interval/rectangle/convex, classic algorithms for proportional cake-cutting, such as \citet{Even1984Note}, can be easily made to return interval/rectangle/convex pieces by ensuring that the cuts are parallel. Similarly, the ownership requirement is easy to satisfy without the geometric constraints, as shown by Theorem \ref{thm:general+}. It is the combination of these two requirements that leads to interesting challenges.
\end{remark}

\rev{
Most land-estates are not exact rectangles, but they can be approximated by a}
\emph{rectilinear polygon}---a polygon in which all angles are $90^\circ$ or $270^\circ$. 
The next result generalizes Theorem \ref{thm:rectangle+} to a rectilinear polygonal cake.
The complexity of a rectilinear polygon is characterized by the number of its \emph{reflex vertices}---vertices with a $270^\circ$ angle. Denote this number by  $T$. A rectangle---the simplest rectilinear polygon---has $T=0$. The cake in Figure \ref{fig:rectilinear} has $T=4$.
\begin{figure}
\begin{center}
\includegraphics[scale=.6]{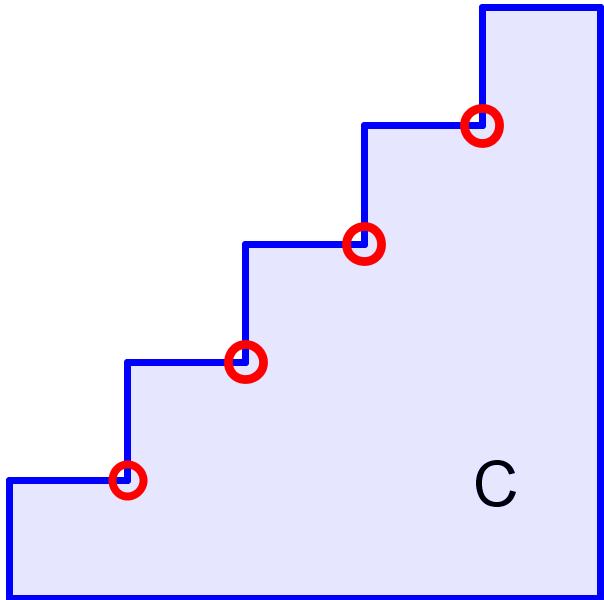}
\end{center}
\caption{\label{fig:rectilinear}
A rectilinear polygon with $T=4$ reflex vertices (circled).}
\end{figure}
\begin{theorem}
\label{thm:rectilinear+}
When the cake is a rectilinear polygon with $T$ reflex vertices,
and each piece must be a rectangle,
for every existing allocation of the cake, 
it is possible to find in time $O(n^2 \log{n}+poly(T))$
a division satisfying democratic-ownership, in which each agent receives at least $1/(3 n+T)$ of the total cake value.%
	\iffalse
\footnote{
The guarantee of $1/(3 n+T)$ is calculated as a fraction of the total cake value. However, with a rectilinear cake and a rectangular piece, even a single agent cannot always get the entire cake value to itself. Therefore, one could think of an alternative guarantee where the benchmark for each agent is the largest value that this agent can attain in a rectangle. For example, one could guarantee each agent a fraction $1/(3 n)$ of the value of their most valuable rectangle. However, such guarantee might be worse than the guarantee of Theorem \ref{thm:rectilinear+}. The proof in the appendix of \citet{segalhalevi2020multicake} implies that the value of the most valuable rectangle might be as small as $1/(T+1)$ of the total cake value. Therefore, the alternative guarantee of $1/(3 n)$ this value translates to a guarantee of $1/(O(n\cdot T))$, which is worse than the $1/(O(n+T))$ guaranteed by Theorem \ref{thm:rectilinear+}.
}
\fi
\end{theorem}
The dependence on $T$ is necessary: even without ownership requirements, there are instances in which it is impossible to guarantee a fraction of more than $1/(n+T)$ to all $n$ agents \citep{segalhalevi2020multicake}.

\subsection{\textbf{Results: Price of Fairness}}
Redivision algorithms can be used not only to compromise between old and new agents, but also to compromise between fairness and efficiency. Often, the most economically-efficient allocation is not fair, while a fair allocation is not economically-efficient. The trade-off between fairness and efficiency is quantified by the \emph{price-of-fairness} \citep{Bertsimas2011Price,Bertsimas2012EfficiencyFairness,Caragiannis2012Efficiency,Aumann2015Efficiency}. It is defined as the worst-case ratio of the maximum attainable social-welfare to the maximum attainable social-welfare of a fair allocation. The social welfare is usually defined as the arithmetic mean of the agents' values (also called \emph{utilitarian welfare}) or their geometric mean (also called \emph{Nash welfare}; see  \citet{Moulin2004Fair}). 

A redivision algorithm can be used to calculate an upper bound on the price of fairness in the following way. Take a welfare-maximizing allocation as the initial allocation; use a redivision algorithm to produce a partially-proportional allocation in which the utility of each agent is close to their initial utility; conclude that the new welfare is close to the initial (maximal) welfare.

Without geometric constraints, the following is an upper bound on the price-of-fairness w.r.t. utilitarian welfare.%
\footnote{
The price-of-fairness of $r$-proportionality  w.r.t. the Nash welfare is $1$ for all $r\leq 1$, since any Nash-optimal cake allocation is proportional \citep{segalhalevi2018monotonicity}.
}
\begin{theorem}
	\label{thm:general-pof}
	For every $r\in[0,1]$, the utilitarian price of $r$-proportionality is at most $1/(1-r)$.
\end{theorem}
When $r=1$ the above bound is infinite, and indeed, the price of 1-proportionality in this setting is $\Theta(\sqrt{n})$, which is not bounded by any constant
\citep{Caragiannis2012Efficiency}.  Theorem \ref{thm:general-pof} shows that a small compromise on the level of proportionality allows a constant (independent of $n$) bound on the utilitarian-price. The parameter $r$ sets the level of trade-off between fairness and efficiency.

With geometric constraints, the following upper bounds are proved:
\begin{theorem}
	\label{thm:interval-pof}
	When the cake is an interval 
	and each piece must be an interval, 
	for every $r\leq 1/2$:
	\begin{itemize}
		\item The utilitarian-price of $r$\-/proportionality is $O(\sqrt{n})$;
		\item The Nash-price of $r$\-/proportionality is at most $5.6$.
	\end{itemize}
\end{theorem}
\begin{theorem}
	\label{thm:rectangle-pof}
	When the cake is a rectangle 
	and each piece must be a rectangle,
	for every $r\leq 1/3$:
	\begin{itemize}
		\item The utilitarian-price of $r$\-/proportionality is $O(\sqrt{n})$;
		\item The Nash-price of $r$\-/proportionality is at most $8.4$.\end{itemize}
\end{theorem}
\begin{theorem}
	\label{thm:convex-pof}
	When the cake is convex polygon 
	and each piece must be convex,
	for every $r\leq 1/4$:
	\begin{itemize}
		\item The utilitarian-price of $r$\-/proportionality is $O(\sqrt{n})$;
		\item The Nash-price of $r$\-/proportionality is at most $11.2$.
	\end{itemize}
\end{theorem}

Note that the first claim in Theorem \ref{thm:interval-pof} is subsumed by \citet{Aumann2015Efficiency}, who prove that the utilitarian-price of 1-proportionality in this setting is $\Theta(\sqrt{n})$. It is brought here only for completeness. The second claim in that theorem, as well as the following theorems regarding two-dimensional constraints, are not implied by previous results.

\rev{
Appendix \ref{sec:pof-lower} partially complements the above results by showing some lower bounds on the price of fairness with interval cake and interval pieces:
\begin{itemize}
\item With two agents, for all $r\in[0,1]$, the utilitarian price of $r$-proportionality is $1+r/2$ and the Nash price of $r$-proportionality 
is $\max(1, \sqrt{2r})$.
\item With $n$ agents, there is a lower bound on the Nash price of proportionality, which approaches $2$ as $n\to\infty$.
\end{itemize}
 Computing the exact utilitarian price and Nash price of $r$-proportionality for any $n\geq 2$ and $r\leq 1$ in this setting remains an open question.
}

\rev{
\begin{remark}
A third measure of welfare is the \emph{egalitarian welfare}, defined as the minimum of the agents' values (normalized such that the total cake value is the same for all agents). 
The egalitarian price of $r$\-/proportionality is $1$ for all $r\leq 1$ whenever an $r$\-/proportional allocation exists. 
This is because, whenever an $r$\-/proportional allocation exists, its egalitarian welfare is at least $r/n$ of the total cake value. Therefore, in any egalitarian-optimal allocation, the egalitarian welfare is at least $r/n$ of the total cake value.
By definition, any such allocation is $r$-proportional. 
So the maximum attainable egalitarian welfare in an $r$-proportional allocation equals the maximum attainable egalitarian welfare overall.
\end{remark}
}

\ifdefined\algorithmForConvexCake
All upper bounds in this paper use redivision algorithms
(In fact, the original motivation for developing the redivision algorithms was to prove upper bounds on price-of-fairness).
%(this is the reason for insisting that the redivision algorithm does not assume that the old division is fair or that all $C$ is divided).
Hence, the upper bounds are constructive in the following sense: their proofs can be used to convert any existing unfair division to a fair division such that the reduction in social welfare is upper-bounded. 
For example, Theorem \ref{thm:interval-pof} can be used to convert any Nash-optimal interval allocation which is not proportional, to a $1/2$\-/proportional interval allocation in which the Nash social welfare is at least $1/5.6$ of the optimum.
\fi

\section{Model}
\label{sec:model}
\subsection{\textbf{Cake Division}}
The \emph{cake} $C$ is a polytope in the $d$-dimensional Euclidean plane $\mathbb{R}^d$. This paper focuses on the common cases in which $d=1$ and $C$ is an interval, or $d=2$ and $C$ is a polygon. A \emph{piece} is a Borel subset of $C$; usually an interval or a polygon.

$C$ has to be divided among $n\geq1$ \emph{agents}. 
We denote by $[n]$ the set of integers $\{1,\ldots,n\}$.
Each agent $i\in[n]$ has a \emph{value-density} function $v_i$, which is an integrable, non-negative and bounded function on $C$. The \emph{value} of a piece $X_i$ to agent $i$ is marked by $V_{i}(X_i)$ and it is the integral of its value-density:
$
V_i(X_i)=\int_{x\in X_i}v_i(x)dx
$.
The definition implies that the $V_i$ are finite measures and are absolutely-continuous with respect to the Lebesgue measure, i.e., any piece with zero area has zero value to all agents. 
%Therefore, there is no need to worry about who gets the boundary of a piece, since its value is 0.

Division algorithms access the value measures via \emph{queries} \citep{Robertson1998CakeCutting,Woeginger2007Complexity}: an \emph{eval query} asks an agent to report the value of a specified piece of cake; a \emph{mark query} asks an agent to mark a piece of cake with a specified value.%
\footnote{
\rev{It is often called a \emph{cut query}, but the term \emph{mark query} better differentiates query answers from actual cuts through the cake.}
}
The present paper ignores strategic considerations and assumes that agents answer truthfully. Indeed, in general it may be impossible to build a cake-cutting algorithm that is both fair and strategy-proof \citep{Branzei2015Dictatorship}.
%As usual in the cake-cutting literature since \citet{Steinhaus1948Problem}, the fairness guarantees of the division algorithms are valid for every agent answering the queries truthfully, regardless of the behavior of the other agents. 
%However, the algorithms are not dominant-strategy-truthful, i.e, an agent may gain by answering untruthfully. Designing dominant-strategy-truthful mechanisms for cake-cutting is known to be a difficult problem even in one dimension, see e.g. \citep{Chen2013Truth,Branzei2015Dictatorship}. Hence, the present paper ignores strategic considerations and assume agents answer truthfully.

The geometric constraints, if any, are represented by a pre-specified family $S$ of \emph{usable pieces}. In this paper, $S$ will either be the set of all pieces (which means that there are no geometric constraints), or the set of all intervals, or the set of all rectangles, or the set of all convex pieces. It is assumed that each agent can use only a single piece from the family $S$.

An \emph{allocation} is a vector of $n$ pieces, $\mathbf{X}=(X_{1},\dots,X_{n})$, one piece per agent, such that the $X_{i}$ are pairwise-disjoint and 
$X_1\sqcup \cdots \sqcup X_n \subseteq C$.%
\footnote{
The symbol $\sqcup$ denotes disjoint union---it emphasizes that the pieces $X_1,\ldots,X_n$ are pairwise-disjoint.
}
 Note that some cake may remain unallocated, i.e, \emph{free disposal} is assumed. As illustrated in the introduction, free disposal may be necessary when there are geometric constraints. An \emph{$S$-allocation} is an allocation in which all pieces are usable, i.e, $X_i\in S$ for each agent $i$.

For every constant $r\in [0,1]$, an allocation $\mathbf{X}$ is called \emph{$r$-proportional} if every agent receives at least $r/n$ of the total cake value:
\begin{align*}
\text{For all~} i\in[n]: &&V_i(X_i) \geq {(r/n)}\cdot V_i(C)
\end{align*}
A 1-proportional division is also known as \emph{proportional}. 

\subsection{\textbf{Cake Redivision}} 
There is an existing $S$-allocation of the cake, $Z_1 \sqcup \dots \sqcup Z_n \subseteq C$. It is assumed that the old pieces $Z_j$ are pairwise-disjoint and that $Z_j\in S$ for all $j$, but nothing else is assumed on the division. In particular, the initial division is not necessarily proportional, and some of $C$ may be unallocated.

It is required to construct a new $S$-allocation
$X_1 \sqcup \cdots \sqcup X_n \subseteq C$. 
The re-allocation satisfies the \emph{$w$-ownership} property, for some constant $w\in(0,1)$, if every agent receives 
%a subset of their old piece worth 
at least a fraction $w$ of their old value:

\begin{align*}
\text{For all~} j\in[n]:
% && 
% X_j\subseteq Z_j 
% && 
%\text{ and } 
&& 
V_j(X_j) \geq w \cdot V_j(Z_j)
\end{align*}
Since $w$-ownership is not always compatible with $r$-proportionality for any constant $r>0$, the following weaker property is defined. A re-allocation $\mathbf{X}$ satisfies the \emph{democratic-ownership} property if, for every $d\range{1}{n-1}$, there are at least $n-d$ agents $j\in[n]$ for whom

\begin{align*}
&&V_j(X_j) > \frac{1}{\lceil n/d\rceil} \cdot V_j(Z_j).
\end{align*}

\subsection{\textbf{Social Welfare and Price-of-Fairness}}
In addition to fairness, it is often required that a division has a high \emph{social welfare}. The social welfare of an allocation is a certain aggregate function of the normalized values of the agents (the normalized value is the piece value divided by the total cake value). 
Common social welfare functions are 
%egalitarian, 
sum (utilitarian) and product (Nash), see \cite{Moulin2004Fair}. 
\rev{When calculating the welfare, it is convenient to normalize the values such that
the proportional share of an agent corresponds to a value of $1$
(so  receiving the entire cake corresponds to a value of $n$). This way, when all agents receive exactly their proportional share, the welfare is $1$.}
\begin{itemize}
	\item \emph{Utilitarian welfare}---the arithmetic mean of the agents' normalized values
	\begin{align*}
	W_{util}(\mathbf{X}) = \frac{1}{n}\sum_{i=1}^n \frac{V_i(X_i)}{V_i(C)/n}
	\end{align*}
	\item \emph{Nash welfare}---the geometric mean of the agents' normalized values:
	
	\begin{align*}
	W_{Nash}(\mathbf{X}) = \left(\prod_{i=1}^n \frac{V_i(X_i)}{V_i(C)/n}\right)^{1/n}
	\end{align*}
	\end{itemize}
	The goal of maximizing the social welfare is not always compatible with the goal of guaranteeing a fair share to every agent. For example, \citet{Caragiannis2012Efficiency} describe a simple example in which the maximum utilitarian welfare of a proportional allocation is in $1$ while the maximum utilitarian welfare of an arbitrary (unfair) allocation is in $\Omega(\sqrt{n})$. This means that society has to pay a price, in terms of social-welfare, for insisting on fairness. This is called the \emph{price of fairness}. Formally, given a social welfare function $W$ and a fairness criterion $F$, the price-of-fairness relative to $W$ and $F$ (also called: ``the $W$-price-of-$F$'') is the ratio:
	
	\begin{align}\label{eq:pof}
	\frac{\sup_{\mathbf{X}} W(\mathbf{X})}{\sup_{\mathbf{Y}\in F} W(\mathbf{Y})} \tag{*}
	\end{align}
	where the supremum at the numerator is over all allocations $\mathbf{X}$ and the supremum at the denominator is over all allocations $\mathbf{Y}$ that also satisfy the fairness criterion $F$. The cited example shows that the utilitarian\-/price\-/of\-/proportionality is in $\Omega(\sqrt{n})$. 
	
	When there are geometric constraints, they affect both the numerator and the denominator of (*), i.e, the suprema are taken only on $S$-allocations. Therefore, it is not a-priori clear whether the price-of-fairness with constraints is higher or lower than without constraints.

\setcounter{proposition}{0}
\setcounter{theorem}{0}

\section{Arbitrary Cake and Arbitrary Pieces}
\label{sec:general}
In this section there are no geometric constraints on the cake or its pieces.
Consider the negative result first.

\rev{
\begin{proposition}
%[\ref{thm:general-}]
For every constants $r,w\in[0,1]$ where $r+w > 1$, it may be impossible to simultaneously guarantee $r$-proportionality and $w$-ownership.
\end{proposition}
}
\begin{proof}
Here is a scenario in which no $r$-proportional division satisfies $w$-ownership.
In the initial allocation, a single agent owns the entire cake. All $n$ agents have the same value-density and they value the entire cake at $1$. In any $r$-proportional division, the $n-1$ landless citizens must receive a total value of $(n-1)r/n = r-r/n$. Therefore the old landlord receives at most $1-r+r/n$. By assumption, $1-r<w$. Hence, if $n$ is sufficiently large, the old landlord receives less than $w$ of his/her previous value, contradicting $w$-ownership.
\qed
\end{proof}

The proof of the matching positive result requires a lemma.
\begin{lemma}\label{lem:convexity}
Given cake-allocations $\mathbf{Z}$ and $\mathbf{Y}$ and
a constant
$r\in[0,1]$, 
there exists an allocation $\mathbf{X}$ such that, for every agent $i\in[n]$:~~
$
V_i(X_i) \geq r V_i(Y_i) + (1-r) V_i(Z_i)
$.
Moreover, when $r$ is a constant rational number, $\mathbf{X}$ 
can be found using $O(n^2 \cdot \operatorname{len}(r))$ queries, where $\operatorname{len}(r)$ is the length of the binary representation of $r$.
\end{lemma}
\begin{proof}
Let us begin with an existential proof. Consider the set of all possible cake-partitions. For each cake-partition, consider the $n\times 1$ vector of utilities of the agents. The Dubins--Spanier theorem \citep{Dubins1961How} implies that the set of all such vectors is convex. So there is an allocation $\mathbf{X}$ satisfying the requirement as an equality: $\forall i\in[n]: V_i(X_i) = r V_i(Y_i) + (1-r) V_i(Z_i)$. 
	
The Dubins--Spanier theorem is not constructive. 
But when $r$ is a rational number, $r=p/q$ with $p<q$ some positive integers, 
an allocation $\mathbf{X}$ satisfying the lemma requirements can be constructed in polynomial time using an algorithm for a different problem:
\emph{fair division with different entitlements}.
In this problem, each agent $i\in [n]$
is entitled to a share $d_i/D$ of the entire cake, where $d_i$ is a positive integer and $D = \sum_i d_i$. Recently, \citet{cseh2020complexity} presented an algorithm that finds, using $2(n-1)\lceil\log_2(D)\rceil$ queries,
an allocation in which the value of each agent $i$ is at least $d_i/D$ of the total cake value. 
This algorithm should be applied to all pairs of agents. 
For every pair $i, j\in [n]$ with $i\neq j$, 
partition $Y_i\cap Z_j$ between $i$ and $j$ with $d_i := p$ and $d_j := q-p$.
The pairs $i,j$ can be processed in any order, even in parallel.
Finally, agent $i$ also gets the entire piece $Y_i\cap Z_i$.

For the sake of the proof, divide this latter piece arbitrarily into two subsets: one is worth $\frac{p}{q}V_i(Y_i\cap Z_i)$ and the other $\frac{q-p}{q}V_i(Y_i\cap Z_i)$ for agent $i$.
Now, each agent $i$ is allocated a piece $X_i$ which can be written as a union of $2 n$ disjoint subsets:
some $n$ subsets of $Y_i\cap Z_1,\ldots, Y_i\cap Z_n$, and some $n$ subsets of $Y_1\cap Z_i,\ldots, Y_n\cap Z_i$.
The former subsets are worth for $i$ at least 
$\frac{p}{q}V_i(Y_i\cap Z_1) + \cdots + \frac{p}{q}V_i(Y_i\cap Z_n)  = \frac{p}{q}V_i(Y_i\cap C)= \frac{p}{q}V_i(Y_i) = r V_i(Y_i)$,
and the latter subsets are worth for $i$ at least 
$\frac{q-p}{q}V_i(Y_1\cap Z_i) + \cdots + \frac{q-p}{q}V_i(Y_n\cap Z_i) = \frac{q-p}{q}V_i(Z_i\cap C) = \frac{q-p}{q}V_i(Z_i) = (1-r)V_i(Z_i)$.
%Therefore $V_i(X_i)\geq r V_i(Y_i) + (1-r)V_i(Z_i)$ as claimed.

The algorithm requires $O(\log{q})$ steps for every pair and $O(n^2 \log{q})$ steps overall.
\rev{Since $\operatorname{len}{(r)} = \log_2{p}+\log_2{q}\geq \log_2{q}$}, the run-time is 
in $O(n^2 \operatorname{len}{(r)})$ as claimed.
\qed\end{proof}

\begin{algorithm}
\caption{
\label{alg:general+}
Cake allocation with partial proportionality and ownership.
}
\begin{algorithmic}[1]
\REQUIRE ~\\
\begin{itemize}
\item A cake $C$ and an existing allocation $Z_1\sqcup\cdots\sqcup Z_n \subseteq C$;
\item A rational number $r = p/q$.
\end{itemize}
\ENSURE A new allocation $X_1\sqcup\cdots\sqcup X_n\subseteq C$ satisfying the following:
\begin{itemize}
\item Partial proportionality: for all $i\in[n]$, $ V_i(X_i)\geq r\cdot V_i(C)/n$.
\item Partial ownership: for all $i\in[n]$, $V_i(X_i)\geq (1-r)\cdot V_i(Z_i)$.
\end{itemize}
\STATE Find a proportional allocation  $Y_1\sqcup\cdots\sqcup Y_n = C$.
\FOR{$i:=1,\ldots,n$ and $j:=1,\ldots,n$ (when $i\neq j$)}
\STATE 
Using an algorithm for fair cake cutting with different entitlements  \citep{cseh2020complexity}, 
divide $Y_i\cap Z_j$ between $i$ and $j$ such that $i$ is entitled to $p/q$ and $j$ is entitled to $(q-p)/q$.
\ENDFOR
\STATE Allocate to each agent $i\in[n]$ the union of the following pieces:
\begin{itemize}
\item The piece $Y_i\cap Z_i$;
\item $i$'s share from $Y_i\cap Z_j$ for all $j\neq i$;
\item $i$'s share from $Y_j\cap Z_i$ for all $j\neq i$.
\end{itemize}
\end{algorithmic}
\end{algorithm}

\begin{theorem}
%\label{thm:general+}
For every 
constants 
$r,w\in[0,1]$ where $r+w\leq 1$, and for every existing division of the cake, there exists a division that simultaneously satisfies $r$-proportionality and $w$-ownership. 
Moreover, 
when $r,w$ are  rational numbers, such a division 
can be found using $O(n^2\operatorname{len}(r))$ queries, where $\operatorname{len}(\cdot)$ denotes binary representation length.
\end{theorem}

\begin{proof}
Given a pair $r,w$ where $r+w\leq 1$, apply Lemma \ref{lem:convexity}, with
the initial allocation as $\mathbf{Z}$, 
and any proportional allocation as $\mathbf{Y}$ (a proportional allocation can be found efficiently by classic algorithms such as \citet{Steinhaus1948Problem}, \citet{Even1984Note}). 
By Lemma \ref{lem:convexity},
the new division satisfies $r$-proportionality and $(1-r)$-ownership, and $1-r\geq w$.
The process is summarized as Algorithm \ref{alg:general+}.
\qed\end{proof}

\begin{remark}
\label{rem:subset}
(a)
The redivision algorithm gives each agent a piece that is not only worth at least $(1-r) V_i(Z_i)$, but also a subset of $Z_i$ (in addition to a subset of $Y_i$). This may be desirable in some cases. E.g. in land division, old landlords may want not only a high value but also a subset of their old plot.

(b) 
\citet{cseh2020complexity} present an algorithm for cake-cutting even when the entitlements are irrational. The number of queries is finite (but unbounded).
This algorithm can be used in Algorithm \ref{alg:general+} to attain $r$-proportionality even when $r$ is irrational, though it is unclear why any government would be interested in such a strange fairness condition.
\end{remark}

\section{Interval Cake and Interval Pieces}
\label{sec:connected}
In this section the cake is an interval and each piece must be an interval.
Consider the negative result first.
\rev{
\begin{proposition}
%\label{thm:connected-}
When the cake is a 1-dimensional interval and each piece must be an interval, for every positive constants $r,w\in(0,1)$, it may be impossible to simultaneously satisfy $r$-proportionality and $w$-ownership. 

Moreover, for every $r\in(0,1]$ and every integer $d\in[n]$, there might be $d$ agents who, in any $r$-proportional division, receive at most a fraction $1/\lfloor{\frac{n}{d}}\rfloor$ of their old value.
\end{proposition}
}
\begin{proof}
Consider an existing allocation $\mathbf{Z}$, a positive constant $r\in(0,1]$, and an integer $d\leq n$.
Here is a scenario in which, in every $r$-proportional allocation, 
there are $d$ agents $j$ who receive a value of at most 
$V_j(Z_j)/\lfloor{n\over d}\rfloor$.
~~
\rev{
Partition the set of $n$ agents into $d$ groups:
\begin{enumerate}
\item $(n\mod d)$ groups containing $\lceil{n\over d}\rceil$ agents; these groups exist only when $d$ does not divide $n$.
\item $(d - n\mod d)$ groups containing $\lfloor{n\over d}\rfloor$ agents.
\end{enumerate}
Note that the total number of agents in these groups is indeed
$(n \mod d)\cdot \lceil \frac{n}{d}\rceil + (d - n \mod d)\cdot \lfloor \frac{n}{d}\rfloor = n$.

In each group of type 1, a single agent $j$ has a nonempty share $Z_j$ in the initial allocation. Agent $j$ values $Z_j$ at $\lceil{n\over d}\rceil$, and the rest of the cake at $0$. The value-density inside $Z_j$ is piecewise-uniform: it has $\lceil{n\over d}\rceil$ regions with a value of $1$ and $\lceil{n\over d}\rceil-1$ ``gaps''---regions with a value of 0. Each of the other $\lceil{n\over d}\rceil-1$ agents in the same group assigns a positive value only to a unique gap in $Z_j$; Figure \ref{fig:n-fractioned-value-density} illustrates the value-densities that are positive in one such $Z_j$. 

In each group of type 2, the valuations are defined similarly to the groups of type 1, except that $j$'s total value is $\lfloor{n\over d}\rfloor$ and there are $\lfloor{n\over d}\rfloor-1$ other agents.

\begin{figure}[h]
\begin{center}
\includegraphics[width=.8\columnwidth]{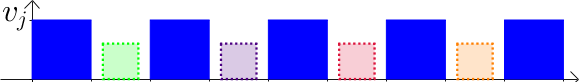}
\caption{
\label{fig:n-fractioned-value-density}
Solid boxes represent the value-density of agent $j$ within $Z_j$; each dotted box represents a value-density of some other agent in the same group as agent $j$. In this example, $\lceil {n\over d}\rceil = 5$.
}
\end{center}
\end{figure}
In any $r$-proportional division, each gap in $Z_j$ must be at least partially allocated to an agent in group $j$. Hence, the interval allocated to agent $j$ must contain at most a single positive region in $Z_j$---it cannot overlap any gap. Therefore the value of agent $j$ is at most $1$. 
For agents in groups of type 2, this value is at most $V_j(Z_j)/\lfloor{n\over d}\rfloor$; for agents in groups of type 1, it is at most 
$V_j(Z_j)/\lceil{n\over d}\rceil$, which is even smaller.
}
\qed
\end{proof}

The corresponding positive result (Theorem \ref{thm:connected+}) uses an algorithm for a different problem: \emph{fair multicake cutting}. 
In this problem,
there is a multicake $C$, which is a union of $m$ pairwise-disjoint subcakes, $C = Z_1 \sqcup \cdots \sqcup Z_m$.
The goal is to give each agent 
a piece contained in a single subcake.
It is easy to see that a proportional allocation might not exist even for a single agent.  
However, there always exists an allocation $(X_1,\ldots,X_n)$ such that

\begin{align}
\label{eq:multicake}
V_i(X_i) \geq 
%\min\left(\frac{1}{n},\frac{k}{m+n-1}\right) 
\frac{1}{m+n-1}
\cdot V_i(C),
\end{align}
and this is the largest fraction that can be guaranteed \citep{segalhalevi2020multicake}.%
\footnote{
Consider $m$ subcakes and $n$ agents with the same valuation, who value the entire multicake at $m+n-1$.
Suppose the value of each subcake $j$ (for all agents) is some integer $u_j \geq 1$.
If some subcake is not allocated to any agent, then the multicake can be reduced to a smaller one with $m' := m-1$ subcakes, which all agents value at most $m'+n-1$. So suppose each subcake is allocated to at least one agent. 
Define the \emph{surplus} of each subcake as $u_j$ minus the number of agents who are allocated a piece in that subcake. The total surplus is $(m+n-1)-n = m-1$, so at least one subcake $j_0$ must have a surplus of at most $0$. At least one of the $u_{j_0}$ agents allocated a piece in subcake $j_0$ has a value of at most $1$.
}
Below, a different algorithm is presented, that attains the same value guarantee \eqref{eq:multicake}, and simultaneously guarantees democratic ownership.

\begin{algorithm}
\caption{
\label{alg:auction}
Auctioning a cake.
}
\begin{algorithmic}[1]
\REQUIRE A cake $Z_0$ 
and a set $N\subseteq [n]$ of agents.
\ENSURE A subset of agents $W\subseteq N$, \rev{possibly empty,} such that
\begin{align*}
(a) && \text{For each agent } & i'\in W, && V_{i'}(Z_0)\geq |W|;
\\
(b) && \text{For each agent } & i''\in N\setminus W, && V_{i''}(Z_0) < |W|+1.
\end{align*}
\STATE Choose an ordering $\sigma$ on $N$ such that 
$V_{\sigma [1]}(Z_0) \geq V_{\sigma [2]}(Z_0) \geq \cdots V_{\sigma [|N|]}(Z_0)$.
\STATE Initialize $W := \emptyset$.
\FOR{$j := 1, \ldots, |N|$}
\IF{$V_{\sigma [j]}(Z_0)\geq j$}
\label{step:if}
\STATE Add agent $\sigma [j]$ to $W$.
\ELSE
\RETURN $W$.
\ENDIF
\ENDFOR
\end{algorithmic}
\end{algorithm}

The algorithm uses as a subroutine Algorithm \ref{alg:auction}, which is called an ``auction''. 
It accepts as input a subcake $Z_0$ and a set $N$ of agents.
Each agent $i$ ``bids'' by evaluating $Z_0$.
The auction then chooses a subset  $W\subseteq N$ of ``winners''.
\rev{
The criterion for selecting the set of winners is specified by the following lemma.
\begin{lemma}
\label{lem:auction}
Given a subcake $Z_0$ and a set $N$ of agents, 
Algorithm \ref{alg:auction}  returns a subset $W\subseteq N$ of winners 
such that (a) each winner values $Z_0$ at least $|W|$, and 
(b) each loser values $Z_0$ at less than $|W|+1$.
\end{lemma}
\begin{proof}
The set of winners contains the first $|W|$ agents in the ordering $\sigma$. Now:
\begin{itemize}
\item 
Let $i' := \sigma[|W|] = $ the last agent added to $W$.
Step \ref{step:if} implies that $V_{i'}(Z_0)\geq |W|$. The same is true for all preceding agents in the ordering $\sigma$. Hence, condition (a) is satisfied for all winners.
\item 
Let $i'' := \sigma [|W|+1] = $ the first agent \emph{not} added to $W$.
Step \ref{step:if} implies that $V_{i''}(Z_0) < |W|+1$. The same is true for all following agents in the ordering $\sigma$. Hence, (b) is satisfied for all losers.\qed
\end{itemize}
\end{proof}
Note that Lemma \ref{lem:auction} allows the set of winners $W$ to be empty, if all agents in $N$ value $Z_0$ at less than $1$.
}

Before proving Theorem \ref{thm:connected+}, let us consider a simpler warm-up algorithm that attains only the partial-proportionality guarantee \eqref{eq:multicake}. 
It uses Algorithm \ref{alg:assign-to-multicake}. Its input is a multicake and a set of $n$ agents. By repeatedly applying the auction algorithm, it assigns the agents to the subcakes
\rev{
such that all agents assigned to a subcake value it sufficiently high, as formalized below.
\begin{lemma}
\label{lem:assign-to-multicake}
Given a multicake  $C = Z_1 \sqcup \cdots \sqcup Z_m$, a positive integer $n\leq m$,
and some $n$ agents who value $C$ at least $m+n-1$,
Algorithm \ref{alg:assign-to-multicake}
returns a partitioning of the set of agents $[n] = W_1\sqcup \cdots \sqcup W_m$ such that
for all $j\in[m]$ and for each agent $i'\in W_j$, 
$V_{i'}(Z_j)\geq |W_j|$.
\end{lemma}
\begin{proof}
\rev{Lemma \ref{lem:auction}(a)} ensures that all agents assigned to $W_j$ in step \ref{step:auction3} value $Z_j$ at least $|W_j|$.
It remains to prove that, by the end of the algorithm, every agent $i \in[n]$ is assigned to some $W_j$.
Suppose by contradiction that some $i\in[n]$ is not in any $W_j$. \rev{Lemma \ref{lem:auction}(b)} ensures that $V_i(Z_j) < |W_j|+1$ for all $j\in[m]$. Summing over all $j\in[m]$ gives $V_i(C) = V_i(Z_1)\cup\cdots \cup V_i(Z_m) < \sum_{j=1}^m (|W_j|+1) \leq (n-1)+m$, which contradicts the assumption $V_i(C) \geq (n-1)+m$.
\qed
\end{proof}
}

Once the agents are partitioned using Algorithm \ref{alg:assign-to-multicake}, for each $j\in[m]$, $Z_j$  can be divided among the agents in $W_j$ using any proportional cake-cutting algorithm. Since all these agents value $Z_j$ at least $|W_j|$, each agent gets a piece valued at least $1$, which is at least $\frac{1}{m+n-1}\cdot V_i(C)$ as in condition \eqref{eq:multicake}.

\begin{algorithm}
\caption{
\label{alg:assign-to-multicake}
Assigning agents to a multicake (warm-up algorithm).
}
\begin{algorithmic}[1]
\REQUIRE A multicake $C = Z_1 \sqcup \cdots \sqcup Z_m$ 
and a set $[n]$ of agents\rev{, where $m\geq n$}.
\begin{itemize}
\item Valuations are normalized such that $V_i(C) = m+n-1$ for all $i\in[n]$.
\end{itemize}

\ENSURE A partitioning of the agents $[n] = W_1\sqcup \cdots \sqcup W_m$ such that:
\begin{itemize}
\item For all $j\in[m]$ and
for each agent $i\in W_j$, 
$V_{i}(Z_j)\geq |W_j|$.
\end{itemize}
\STATE Initialize $N := [n] = $ the set of all agents.
\FOR{$j := 1, \ldots, m$}
\STATE 
\label{step:auction3}
Using Algorithm \ref{alg:auction}, 
auction the subcake $Z_j$ among the agents in $N$. 
\\
Let $W_j$ be the set of winners.
\STATE Remove $W_j$ from $N$.
\ENDFOR
\end{algorithmic}
\end{algorithm}

To prove Theorem \ref{thm:connected+}, it is required to guarantee, in addition to \eqref{eq:multicake}, also the democratic ownership condition. To this end, Algorithm \ref{alg:assign-to-multicake} is replaced with a modified assignment algorithm, presented as Algorithm \ref{alg:assign-to-multicake-democratic}. 
The main difference is that Algorithm \ref{alg:assign-to-multicake-democratic} allows each agent $j$ to participate in the auction on \rev{the subcake with the same index} $Z_j$, even if $j$ was already assigned to a previous subcake $Z_{j'}$ for some $j'<j$.
If $j$ is one of the winners for $Z_j$ (that is, $j\in W_j$), then $j$ is removed from the previous assignment $W_{j'}$. 
This creates a ``vacancy'' in $W_{j'}$;
this vacancy is filled by running a single  step of the auction on $Z_{j'}$. 
Let $i'$ be the first unassigned agent who did not win the first auction on  $Z_{j'}$ (``first'' by the $\sigma$ ordering in that auction). Recall that, by condition (b) of the auction algorithm,
$V_{i'}(Z_{j'}) < |W_{j'}|+1$ held before $j$ was removed from $W_{j'}$. 
If the condition does \emph{not} hold after $j$ is removed (that is: if $V_{i'}(Z_{j'}) \geq  |W_{j'}|+1$ after the removal), then $i'$ is added to $W_{j'}$.
This step guarantees that both conditions (a) and (b) still hold for $W_{j'}$, that is: $V_{i'}(Z_{j'}) \geq |W_{j'}|$ for all $i'\in W_{j'}$, and $V_{i''}(Z_{j'}) < |W_{j'}|+1$  for all $i''$ who are not assigned yet.

This new winner $i'$, who is added to $W_{j'}$, might be the agent $j'$ itself, who is already assigned to another set $W_{j''}$.
In this case, moving the agent $j'$ from $W_{j''}$ to $W_{j'}$ creates a vacancy in $W_{j''}$, which has to be filled in the same way. This chain reaction must eventually end, since whenever a vacancy is created, the number of agents $j$ who are assigned to \rev{the subset with the same index $Z_j$} increases by one, and this number never decreases as no agent $j$ is ever removed from $W_j$.

\rev{
The correctness of Algorithm \ref{alg:assign-to-multicake-democratic} is proved formally below.}

\begin{algorithm}
\caption{
\label{alg:assign-to-multicake-democratic}
Assigning agents to a multicake, with ownership.
}
\begin{algorithmic}[1]
\REQUIRE A multicake $C = Z_1 \sqcup \cdots \sqcup Z_m$ 
and a set $[n]$ of agents\rev{, where $m\geq n$}.
\begin{itemize}
\item Valuations are normalized such that $V_i(C) = m+n-1$ for all $i\in[n]$.
\end{itemize}

\ENSURE A partitioning of the agents $[n] = W_1\sqcup \cdots \sqcup W_m$ such that
\begin{itemize}
\item (a) 
For all $j\in[m]$ and for each agent $i'\in W_j$, 
$V_{i'}(Z_j)\geq |W_j|$.
\item (b)
For each agent $j\in[n]$, 
either $j\in W_j$
or $V_j(Z_j) < |W_j|+1$.
\end{itemize}
\STATE Initialize $N := [n] = $ the set of all agents.
\FOR{$j := 1, \ldots, m$}
\STATE If $j \leq n$ then let $N_j := N \cup \{j\}$; else 
let $N_j := N$.
\STATE
\label{step:auction}
Using Algorithm \ref{alg:auction}, 
auction the subcake $Z_j$ among the agents in $N_j$. 
\\
 Let $W_j$ be the set of winners.
\STATE
 Remove $W_j$ from $N$.
\IF{$j\in W_j$ and also $j\in W_{j'}$ for some 
$j'<j$}
\STATE 
\label{step:remove}
Remove $j$ from $W_{j'}$.
\STATE 
\label{step:one-step-auction}
Let $i'$ be the first loser in the auction on $Z_{j'}$, who is still in $N\cup \{j'\}$.
\IF{$V_{i'}(Z_{j'})\geq |W_{j'}|+1$}
\STATE 
\label{step:add}
Add $i'$ to $W_{j'}$;
\STATE Remove $i'$ from $N$.
\STATE 
\label{step:remove-again}
If $i' = j'$ and also $i' \in W_{j''}$ for some $j''<j$,
then repeat steps \ref{step:remove}--\ref{step:remove-again} with $j'$.
\ENDIF
\ENDIF
\ENDFOR
\end{algorithmic}
\end{algorithm}

\rev{
\begin{lemma}
\label{lem:assign-to-multicake-democratic}
Given a multicake  $C = Z_1 \sqcup \cdots \sqcup Z_m$, a positive integer \rev{$n\leq m$},
and some $n$ agents who value $C$ at least $m+n-1$,
Algorithm \ref{alg:assign-to-multicake-democratic}
returns a partitioning of the set of agents $[n] = W_1\sqcup \cdots \sqcup W_m$ such that
(a) For all $j\in[m]$ and for each agent $i'\in W_j$, 
$V_{i'}(Z_j)\geq |W_j|$;
(b) For all $j\in[n]$, either $j\in W_j$ or $V_j(Z_j) < |W_j|+1$.
\end{lemma}
\begin{proof}
(a)
\rev{Lemma \ref{lem:auction}(a)} ensures that,
in step \ref{step:auction} in iteration $j$,
all agents assigned to $W_j$ value $Z_j$ at least $|W_j|$.
The same holds for agents added to $W_j$ in step \ref{step:add} in a later iteration.
The size of $W_j$ never increases above its initial level: 
it can only increase by one in step \ref{step:add} after it has decreased by one in step \ref{step:remove}. Therefore, the condition still holds when the algorithm ends.

(b)
Agent $j$ always participates in the auction 
in step \ref{step:auction} of iteration $j$.
There are two possible cases.
\begin{itemize}
\item If $j$ is assigned to $W_j$, then he remains in $W_j$ until the end of the algorithm, since $j$ is never removed from $W_j$.
\item Otherwise, \rev{Lemma \ref{lem:auction}(b)} 
ensures that $V_j(Z_j) < |W_j|+1$. 
It remains to show that the condition still holds in later iterations.

For clarity, focus on a specific agent $j'$, and suppose that in iteration $j'$, agent $j'$ did not win the auction, so
$V_{j'}(Z_{j'}) < |W_{j'}|+1$.
In later iterations, 
the size of $W_{j'}$ may decrease by one in step \ref{step:remove}. 
In this case, lines \ref{step:one-step-auction}--\ref{step:add} guarantee that, either $j'$ is added to $W_{j'}$, 
or another agent is added to $W_{j'}$;
in the latter case, $W_{j'}$ returns to its original size,
so $V_{j'}(Z_{j'}) < |W_{j'}|+1$ still holds.
\end{itemize}

It remains to prove that, by the end of the algorithm, every agent $i \in[n]$ is assigned to some $W_j$.
Suppose by contradiction that some $i\in[n]$ is not in any $W_j$.
This means that $i$ did not win the auction in step \ref{step:auction} in any iteration $j$,
so Condition (b) of the auction algorithm ensures that $V_i(Z_j) < |W_j|+1$ for all $j\in[m]$.

For any $j'$, the size of $W_{j'}$ may decrease by one in step \ref{step:remove} in later iterations.
In this case, lines \ref{step:one-step-auction}--\ref{step:add} guarantee that, if $i$ is not added to $W_{j'}$, 
then 
either $V_{i}(Z_{j'}) < |W_{j'}|+1$  still holds,
or another agent is added to $W_{j'}$;
in the latter case, $W_{j'}$ returns to its original size,
so $V_{i}(Z_{j'}) < |W_{j'}|+1$ still holds.

Summing over all $j\in[m]$ gives $V_i(C)  < \sum_{j=1}^m (|W_j|+1) \leq (n-1)+m$, contradicting the assumption $V_i(C) \geq (n-1)+m$. \qed
\end{proof}
}

The above lemmas and algorithms are used to prove the following theorem.
\begin{algorithm}
\caption{
\label{alg:connected+}
Allocation of an interval cake, with partial proportionality and democratic ownership.
}
\begin{algorithmic}[1]
\REQUIRE ~\\
\begin{itemize}
\item An interval $C$ and an existing allocation into $n$ intervals, $Z_1 \sqcup \cdots \sqcup Z_n\subseteq C$.
\end{itemize}
\ENSURE A new allocation  $X_1 \sqcup \cdots \sqcup X_n \subseteq C$ into $n$ intervals, satisfying the following:
\begin{itemize}
\item $1/2$-proportionality: for all $i\in[n]$, $ V_i(X_i)\geq V_i(C)/(2 n)$.
\item Democratic ownership: for all $d<n$, for at least $n-d$ agents $i$,  $V_i(X_i)> V_i(Z_i) / \lceil\frac{n}{d}\rceil$.
\end{itemize}
\STATE 
Normalize the valuations such that $V_i(C) = 2 n-1$ for all $i\in[n]$.

\STATE 
\label{step:completion}
Given the original partial allocation $Z_1 \sqcup \cdots \sqcup Z_n \subseteq C$, expand it to a complete partition $Z_1' \sqcup \cdots \sqcup Z_n'  = C$, by attaching each ``blank'' (unallocated interval in $C$) arbitrarily to one of the two adjacent allocated intervals, to its left or to its right. 

\STATE
\label{step:multicake-division}
Considering the intervals $Z'_1,\dots,Z'_n$ as subcakes in a multicake, 
use Algorithm \ref{alg:assign-to-multicake-democratic} (with $m=n$) 
to partition the agents into $n$ subsets $W_1,\ldots,W_n$.

\STATE
\label{step:proportional-division}
Divide each interval $Z'_j$ among the agents in $W_j$ using any algorithm for connected proportional cake-cutting, e.g. \citet{Even1984Note}.
\end{algorithmic}
\end{algorithm}

\rev{
\begin{theorem}
%\label{thm:connected+}
When the cake is a 1-dimensional interval 
and each piece must be an interval, 
it is possible to find in time $O(n^2 \log{n})$ 
a division simultaneously satisfying democratic-ownership and $1/2$-proportionality.
\end{theorem}
}

\begin{proof}
The proof is constructive and uses Algorithm \ref{alg:connected+}.
It is proved below that the output of this 
algorithm satisfies the requirements of the theorem.

\item 
\paragraph{Proof that the output of Algorithm \ref{alg:connected+} satisfies  $1/2$-proportionality.}
Each subcake $Z'_j$ is divided proportionally among the agents in $W_j$. 
\rev{By Lemma \ref{lem:assign-to-multicake-democratic}(a)}, all these agents value $Z'_j$ at least $|W_j|$. Hence, their piece has a value of at least $1$. By the normalization step (with $m=n$), $1\geq \frac{1}{2n-1}V_i(C) > V_i(C)/(2n)$.

\item 
\paragraph{Proof that the output of Algorithm \ref{alg:connected+} satisfies democratic-ownership.}
Applying the pigeonhole principle to the partition yielded by Algorithm \ref{alg:assign-to-multicake-democratic} 
implies that, for every integer $d \range{1}{n-1}$, at most $d$ of the subsets $W_j$, for $j\in[n]$, are populated by at least $\lceil {n\over d}\rceil$ agents. Hence, at least $n-d$ such subsets are populated by at most $\lceil {n\over d}\rceil-1$ agents, that is, they satisfy $|W_j|\leq \lceil{n\over d}\rceil-1$. 
For each $j\in[n]$, consider two cases:
\begin{description}
\item[\emph{Case \#1}]: $j\in W_j$. 
Then agent $j$ receives a piece of the subcake $Z'_j$. 
By the proportionality of the subcake division 
(Algorithm \ref{alg:connected+}, step \ref{step:proportional-division}):

\begin{align*}
V_j(X_j) \geq 
\frac{V_j(Z'_j)}{|W_j|}
\geq 
\frac{V_j(Z'_j)}{\lceil {n\over d}\rceil-1}
>
\frac{V_j(Z'_j)}{\lceil {n\over d}\rceil}.
\end{align*}
\item[\emph{Case \#2}]: $j\not \in W_j$. Then, 
\rev{by Lemma \ref{lem:assign-to-multicake-democratic}(b)}, 
$V_j(Z'_j) < |W_j|+1 \leq \lceil {n\over d}\rceil$.
Therefore, 
$V_j(Z'_j) / \lceil {n\over d}\rceil < 1$.
As explained in the proof of $1/2$ proportionality, the value of each agent is at least $1$:

\begin{align*}
V_j(X_j) 
\geq 
1
>
\frac{V_j(Z'_j)}{\lceil {n\over d}\rceil}.
\end{align*}
\end{description}

\iffalse
HERE, IT IS POSSIBLE THAT 

\begin{align*}
V_j(X_j) 
<
\frac{V_j(Z'_j)}{\lfloor {n\over d}\rfloor}.
&&
\left\lfloor {n\over d}\right\rfloor
<
V_j(Z'_j) 
< 
\left\lceil {n\over d}\right\rceil.
\end{align*}
\fi

In both cases, agent $j$ receives a value greater than $V_j(Z'_j) / \lceil {n\over d}\rceil$.
The latter ratio is at least $V_j(Z_j) / \lceil {n\over d}\rceil$ since $Z'_j\supseteq Z_j$.

\item 
\paragraph{Run-time complexity of Algorithm \ref{alg:connected+}.}
The auction in Algorithm \ref{alg:auction} requires $O(n \log n)$ queries.
Algorithm \ref{alg:assign-to-multicake-democratic} 
performs $m$ auctions.
Each auction might lead to a sequence of at most $n$ vacancies which require one query each to be filled.
Algorithm Even--Paz requires $O(n \log{n})$ queries, and it is done $m$ times---once for each subcake.
All in all, the run-time is in $O(m n\log{n}) = O(n^2\log{n})$, since $m=n$.
\qed
\end{proof}

\subsection{\textbf{Future Work}}
\label{sub:future1}
\paragraph{Crossing the boundary lines.}
Algorithm \ref{alg:connected+} 
treats each existing piece $Z_j$ as an isolated subcake, and insists that each new piece be entirely contained in an existing piece, i.e, it does not cross the existing division lines. This may be desirable in the context of land division, since it respects the \emph{Uti possidetis juris}
\citep{lalonde2002determining}---
\rev{an international law principle saying that newly-formed sovereign states should retain the internal borders that their preceding dependent area had before their independence.}
However, it also implies that the resulting division can only be $1/2$-proportional and never fully proportional,
as the fraction $\frac{1}{n+m-1}$ in \eqref{eq:multicake} is tight.

Theoretically, it may be possible to improve the proportionality guarantee by devising a different redivision procedure that crosses the existing division lines.
This raises the following open question: what is the highest level of proportionality that is compatible with democratic-ownership?

\paragraph{Several pieces per agent.}
Theorem \ref{thm:general+} allows an unlimited number of pieces per agent,%
\footnote{
In fact, $4n-3$ pieces per agent are sufficient. It is known that, for every pair of agents, an allocation with different entitlements can be attained with two cuts, e.g. \citep{SegalHalevi2018Entitlements}. So each agent receives at most $2$ pieces. In Algorithm \ref{alg:general+}, each agent participates in $2\cdot(n-1)$ such allocation instances, and gets one additional piece.
}
while Theorem \ref{thm:connected+} allows only a single piece per agent. What happens between these extremes?
In particular, if each agent can get $k$ intervals, for some fixed $k\geq 1$, then 
there is an algorithm for dividing a multicake with $m$ subcakes among $n$ agents such that each agent gets at least 
$\min\left(\frac{1}{n},
\frac{k}{m+n-1}\right)$ of the total cake value
\citep{segalhalevi2020multicake}.
However, the algorithm does not guarantee democratic ownership. 
If a similar proportionality guarantee could be attained together with democratic ownership, it could be used in Section \ref{sec:connected} with $m = k\cdot n$ subcakes (since for each agent there could be up to $k$ subcakes in the original division), to get a bound of $\frac{k}{kn+n-1}$, which implies $\frac{k}{k+1}$-proportionality for any $k\geq 1$.

\section{Polygonal Cake and Polygonal Pieces}
\label{sec:rectangle}
In this section the cake is a polygon in $\mathbb{R}^2$.
There is a set $S$ of \emph{usable pieces} (e.g. rectangles), the initial allocation $Z_1,\ldots,Z_n$ is an $S$-allocation, and the output should be an $S$-allocation too.

The main obstacle in applying Algorithm \ref{alg:connected+} to such a cake is step \ref{step:completion}---\rev{extending the initial partial allocation to a complete partition of the entire cake.}
It is not possible to simply attach each unallocated part of $C$ to an allocated $S$-piece, since the result might not be an $S$-piece. 
The initial partial allocation $Z_1\sqcup\cdots \sqcup Z_n \subseteq C$ still must be expanded to a complete partition of $C$,  since Algorithm \ref{alg:connected+} uses  Algorithm \ref{alg:assign-to-multicake-democratic}, which requires a complete partition.
But the number of pieces in the complete partition might be larger than $n$, since there might be unattached ``blanks'' (holes).

The goal, then, is to find a partition of $C$ into $S$-pieces, $Z_1'\sqcup\cdots \sqcup Z_{n+b}'  = C$, with $b\geq 0$, such that every input $S$-piece is contained in a unique output $S$-piece: $\forall j\in[n]:  Z_j \subseteq Z_j'$. The additional $b$ $S$-pieces are called \emph{blanks}. In Step \ref{step:multicake-division}, the multicake will contain $m=n+b$ subcakes.
Hence, the fraction guaranteed to each agent will be $1/(n+m-1) = 1/(2 n + b - 1)$. A smaller value of $b$ translates to a better proportionality guarantee.

An example of the input and output of the allocation-completion step,
when $S$ is the set of rectangles, is shown in Figure \ref{fig:completion}.
In the partial allocation there are $n=4$ rectangles;
in the complete partition there are $m=5$ rectangles.

\begin{figure}[h]
\begin{center}
\includegraphics[width=.4\columnwidth]{completion-rectangle-1.png}
~$\Rightarrow$~
\includegraphics[width=.4\columnwidth]{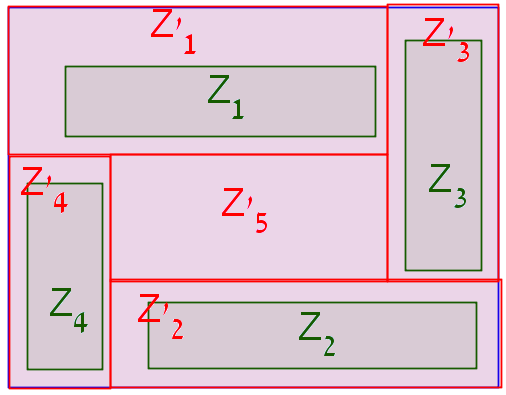}
\end{center}
\caption{
\label{fig:completion}
Allocation-completion with $n=4$ original pieces and $b=1$ blank, denoted $Z_5'$.}
\end{figure}

This raises the question of what is the minimum number of blanks required for a complete partition? 
This geometric question has been studied in a different paper \citep{akopyan2018counting}. 
The answers are summarized in 
Table \ref{tab:completion}.%
\footnote{
The expressions in Table \ref{tab:completion} are tight in the worst case: there are partial allocations that require exactly this number of blanks.
}
Moreover, it is proved there that the worst-case optimal number of blanks is attained in any arrangement in which all pieces are \emph{maximal}, that is, cannot be expanded without overlapping another piece. Formally:
\begin{definition}
Given a set $S$ of usable pieces,
a cake $C$
and some $m$ pairwise-disjoint $S$-pieces $Z_1,\ldots,Z_n \subseteq C$:

(a) An $S$-piece $Z'\subseteq C$ is called \emph{maximal} w.r.t. $C, Z_1,\ldots,Z_n$, if every superset $S$-piece $Z'' \supsetneq Z'$ that is contained in $C$ overlaps one of $Z_1,\ldots,Z_n$.

(b) An $S$-piece $Z'\subseteq C$ is called a \emph{maximal expansion} of $Z_j$ if $Z'\supseteq Z_j$, and $Z'$ is maximal w.r.t. $C, Z_1,\ldots, Z_{j-1}$, $Z_{j+1},\ldots,Z_n$.

(c) A set of pairwise-disjoint $S$-pieces $Z'_1,\ldots,Z'_n \subseteq C$ is called a \emph{complete expansion} of $Z_1,\ldots,Z_n$ if for each $j\in[m]$, $Z'_j$ is a maximal expansion of $Z_j$ w.r.t. $C, Z'_1,\ldots, Z'_{j-1}$, $Z'_{j+1},\ldots,Z'_n$.
\end{definition}

\begin{table}
\caption{
\label{tab:completion}
Worst-case number of blanks in a maximal arrangement of pairwise-disjoint $S$-pieces contained in a cake $C$.
From \citet{akopyan2018counting}.
}
\begin{center}
\begin{tabular}{|c|c|c|}
\hline
\textbf{Cake $C$} &
\textbf{Usable pieces $S$} & \textbf{Number of blanks $b$} \\
\hline
Polygon &
Polygons & $0$ \\
\hline
Simple polygon (without holes) &
Simple polygons & $0$ \\
\hline
Axes-parallel rectangle & 
Axes-parallel rectangles & $n-\lceil 2\sqrt{n}-1 \rceil$ \\
\hline
Convex figure & 
Convex figures & $2n-5$ \\
\hline
Rectilinear polygon, $T$ reflex vertices  & 
Axes-parallel rectangles & $T+n-\lceil 2\sqrt{n}-1 \rceil$ \\
\hline
\end{tabular}
\end{center}
\end{table}

\iffalse
\begin{algorithm}
\caption{
\label{alg:allocation-completion}
Expand a partial allocation to a complete partition.
Steps denoted by (*) require further elaboration (see the text body).
}
\begin{algorithmic}[1]
\REQUIRE A cake $C$ and a partial $S$-allocation $Z_1\sqcup \cdots \sqcup Z_n\subseteq C$.
\ENSURE A complete  $S$-partition $Z'_1\sqcup \cdots \sqcup Z'_{n+b}  = C$, such that
\begin{itemize}
\item $Z'_j\supseteq Z_j$ for all $j\in[n]$;
\item $b$ is upper-bounded as in Table \ref{tab:completion}.
\end{itemize}
\FOR{$j := 1, \ldots, n$}
\STATE \label{step:maximal-expansion}
(*)
Let $Z'_j$ be a maximal expansion of $Z_j$ w.r.t. $C,Z_1,\ldots,Z_{j-1}$,$Z_{j+1},\ldots,Z_n$.
\STATE Replace $Z_j$ with $Z'_j$.
\ENDFOR
~~~~
\COMMENT{\emph{Now, $Z'_1,\ldots,Z'_n$ is a complete expansion of $Z_1,\ldots,Z_n$.}}
\STATE
Let $C'$ be the cake remaining after the expansions, $C' := C\setminus (Z'_1\sqcup \cdots Z'_n)$.
\STATE
\label{step:remainder-partition}
(*)
Let $Z'_{n+1},\ldots,Z'_{n+b}$ be a partition of $C'$ into $b$ $S$-pieces.
\RETURN $Z'_1\sqcup \cdots \sqcup Z'_{n+b}$.
\end{algorithmic}
\end{algorithm}
%Algorithm \ref{alg:allocation-completion} can be used to expand the given partial $S$-allocation $Z_1\sqcup \cdots \sqcup Z_n$ to a complete $S$-partition $Z'_1\sqcup\cdots\sqcup Z'_{n+b}$.
%By construction, the resulting arrangement is maximal, so the number of blanks $b$ is upper bounded by the expressions in Table \ref{tab:completion}.
\fi

\rev{
Complete expansions are used to prove
Theorems \ref{thm:rectangle+}, \ref{thm:convex+} and \ref{thm:rectilinear+} below.
\begin{theorem}
%	\label{thm:rectangle+}
When the cake is a rectangle 
and each piece must be a parallel rectangle,
it is possible to find in time $O(n^2 \log{n})$ 
a division simultaneously satisfying democratic-ownership and $1/3$-proportionality.
\end{theorem}
}

\begin{proof}

\rev{
Find a complete expansion of the initial allocation $Z_1\sqcup\cdots\sqcup Z_n\subseteq C$ in the following way:
}

\begin{framed}
\begin{algorithmic}
\FOR{$j := 1, \ldots, n$}
\STATE
Expand each of the four sides of $Z_j$ until it \rev{touches} the boundary
\\
~~~of $C$ or of another piece $Z_i$ for some $i\neq j$.
\STATE
Denote the expanded piece by $Z'_j$; note that it is  a maximal expansion
\\
~~~ of $Z_j$ 
w.r.t. $C,Z_1,\ldots,Z_{j-1}$,$Z_{j+1},\ldots,Z_n$.
\STATE Replace $Z_j$ with $Z'_j$.
\ENDFOR
\end{algorithmic}
\end{framed}
\rev{
\citet{akopyan2018counting}
prove that the remaining ``holes'' (unfilled parts of $C$) are all rectangular, and their number $b$ satisfies $b \leq n-\lceil 2\sqrt{n}-1 \rceil < n$.
So there is a complete $S$-partition $Z'_1\sqcup\cdots\sqcup Z'_{n+b} = C$.

Considering the $S$-pieces $Z'_1,\ldots,Z'_{n+b}$ as subcakes, use Algorithm \ref{alg:assign-to-multicake-democratic} 
with $m = n+b$
to partition the agents into subsets $W_1\sqcup\cdots\sqcup W_{n+b}$.
Then, use the Even--Paz algorithm to partition each $Z'_j$ among the agents in $W_j$.
While the Even--Paz algorithm was originally presented for 1-dimensional intervals, it is easily applicable to axes-parallel rectangles, for example by ensuring that all cuts are parallel to the $y$ axis. The resulting allocation satisfies democratic ownership. In addition, for each agent $i\in [n]$:
$V_i(X_i)\geq$ 
$\frac{1}{n + m -1}V_i(C)\geq$
$\frac{1}{2n + b-1}V_i(C)>$
$V_i(C)/(3n)$, 
so the allocation is $1/3$-proportional.

The run-time of finding a maximal expansion of a rectangle is in $O(n)$, since it requires to compare each side of the rectangle to the sides of the other $n-1$ rectangles. Therefore, the run-time of finding a complete expansion is $O(n^2)$.
The run-time of executing Algorithm \ref{alg:assign-to-multicake-democratic} and the Even--Paz algorithms is $O(n^2 \log{n})$ as in Theorem \ref{thm:connected+}, so the total run-time is $O(n^2 \log{n})$ too.
}
\qed
\end{proof}

\rev{
\begin{theorem}
%\label{thm:convex+}
When the cake is a 2-dimensional convex polygon 
and each piece must be convex,
there exists a division simultaneously satisfying democratic-ownership and $1/4$-proportionality.
\end{theorem}
}
\begin{proof}

The proof is similar to that of the previous theorem,
and relies on the existence of a complete expansion of the initial allocation. However, 
I do not have a constructive algorithm for finding a maximal expansion of a convex figure.
Recently, 
\citet{Dmitry2021Maximal} presented an algorithm for finding a maximal expansion of a convex polygon contained in an arbitrary polygon. However, since it is not formally published, Theorem \ref{thm:convex+} is stated only as an existence result.

To prove existence of a maximal expansion of $Z_j$,\footnote{The following proof is an adaptation of the proof of Lemma 3.3 in \cite{Mohammadi2017Reconstruction}. I am grateful to Ashkan for his help with this argument.} let $Y_j$ be the set of potential expansions of $Z_j$, i.e:
\begin{align*}
Y_j := \{ 
X_j | 
X_j \supseteq Z_j
\text{~and~}
X_j \subseteq (C \setminus \cup_{i\neq j} Z_i)
\text{~and~}
X_j \text{~is convex.}
\}
\end{align*}
$Y_j$ is partially ordered by inclusion.
The Kuratowski--Zorn lemma can be used to prove that it has a maximal element.
To use this lemma, one has to prove that every chain in $Y_j$ has an upper bound in $Y_j$. Indeed, let $Y'_j \subseteq Y_j$ be a chain. Let $\widehat{Y'_j}$ be the union of all sets in $Y'_j$. Then $\widehat{Y'_j}$ is an upper bound on $Y'_j$, and $\widehat{Y'_j}\in Y_j$ because:
\begin{itemize}
\item $\widehat{Y'_j} \supseteq Z_j$---since all sets in $Y_j$ contain $Z_j$.
\item $\widehat{Y'_j} \subseteq (C \setminus \cup_{i\neq j} Z_i)$---since all sets in $Y_j$ are contained in $(C \setminus \cup_{i\neq j} Z_i)$.
\item $\widehat{Y'_j}$ is convex---since for every two points in $\widehat{Y'_j}$, there exists a set in the chain $Y'_j$ that contains both of them. This set is convex so it contains the segment between them, so $\widehat{Y'_j}$ contains this segment too.
\end{itemize}
Thus, by the Kuratowski--Zorn lemma, $Y_j$ has a maximal element.
Denote this element by $Z_j'$.
By definition, it is a maximal expansion of $Z_j$ w.r.t $C,Z_1,\ldots,Z_{j-1},Z_{j+1},\ldots,Z_n$.
As in the proof of Theorem \ref{thm:rectangle+},
one can proceed iteratively for $j:=1,\ldots,n$, replacing each $Z_j$ with its maximal expansion $Z'_j$. This
yields a complete expansion of $Z_1,\ldots,Z_n$.
\citet{akopyan2018counting} prove that the remaining holes are all convex, and their number $b$ satisfies $b \leq  2n-5 < 2n$.
So there is a complete $S$-partition $Z'_1\sqcup\cdots\sqcup Z'_{n+b} = C$.

Considering the $S$-pieces $Z'_1,\ldots,Z'_{n+b}$ as subcakes, use Algorithm \ref{alg:assign-to-multicake-democratic} 
with $m = n+b$
to partition the agents into subsets $W_1\sqcup\cdots\sqcup W_{n+b}$.
Then, use the Even--Paz algorithm to partition each $Z'_j$ among the agents in $W_j$. Requiring that all cuts made by this algorithm are parallel to the y-axis guarantees that the pieces are convex. 
The resulting allocation satisfies democratic ownership, and for each agent $i\in [n]$, 
$V_i(X_i)\geq$ 
$\frac{1}{n + m -1}V_i(C)\geq$
$\frac{1}{2n + b-1}V_i(C)>$
$V_i(C)/(4n)$, 
so the allocation is $1/4$-proportional.
\qed
\end{proof}

\rev{
\begin{theorem}
%\label{thm:rectilinear+}
	When the cake is an axes-parallel rectilinear polygon with $T$ reflex vertices,
	and each piece must be an axes-parallel rectangle,
	it is possible to find in time $O(n^2 \log{n}+poly(T))$
	a division satisfying democratic-ownership, in which each agent receives at least $1/(3 n+T)$ of the total cake value.%
\end{theorem}
}
\begin{proof}
\rev{
The proof starts similarly to Theorem \ref{thm:rectangle+},
by finding a complete expansion of the initial allocation $Z_1\sqcup\cdots\sqcup Z_n\subseteq C$.
\citet{akopyan2018counting} prove that the remaining ``holes'' are all simply-connected rectilinear polygons,
and that they can be partitioned into at most $b$ rectangles, where
$b \leq n-\lceil 2\sqrt{n}-1 + T \rceil < n + T$.
The partitioning can be done in time $\operatorname{poly}(T)$; see
 \citep{Keil2000Polygon,Eppstein2010GraphTheoretic}.
 
Proceeding as in the previous theorems, 
Algorithm \ref{alg:assign-to-multicake-democratic} 
finds an allocation satisfying democratic-ownership in which, for each agent $i\in [n]$,
$V_i(X_i)\geq$ 
$\frac{1}{n + m -1}V_i(C)\geq$
$\frac{1}{2n + b-1}V_i(C)>$
$V_i(C)/(3n+T)$. 
The run-time of this part is in $O(n^2 \log{n})$ as explained in Theorem \ref{thm:rectangle+}.
\qed
%Step \ref{step:remainder-partition} is needed only when $C$ is a rectilinear polygon and $S$ is the family of rectangles (in the other cases, \citet{akopyan2018counting} prove that all holes remaining in $C'$ are already $S$-pieces).
}
\end{proof}

\subsection{\textbf{Future Work}}
\label{sub:future2}
The results in this section raise several future work questions, which may be of interest to researchers in computational geometry.

\paragraph{Rectangle and convex pieces.}
While the bounds in Table \ref{tab:completion} are worst-case optimal, in specific instances there may be a maximal expansion with fewer blanks. 
What is an efficient algorithm for finding a maximal expansion of a given allocation, which has the smallest number of blanks possible in the given instance?

\paragraph{General polygons.}
When the cake and the pieces are general polygons, or hole-free (simply-connected) polygons, but not necessarily convex, 
there exists a maximal expansion with no blanks at all (Table \ref{tab:completion}).
\rev{
Using such an expansion, one could expect to 
find an allocation satisfying democratic ownership and $1/2$-proportionality.
However, this requires to apply the Even--Paz algorithm to a non-convex polygon such that the pieces remain connected (or simply-connected);
cutting along the y axis (as in the rectangle and convex cases)  might yield disconnected pieces.
One way to partition a polygon into connected pieces is to map each point $p$ of the polygon to the point nearest to $p$ on the polygon perimeter, as in Figure \ref{fig:medial-axis}.
\begin{figure}
\begin{center}
\includegraphics[scale=.8]{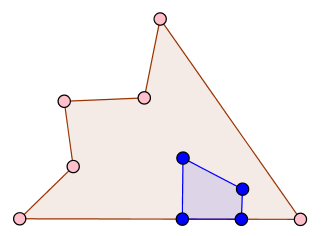}
\end{center}
\caption{\label{fig:medial-axis}
\rev{
The points in the blue trapezoid are mapped to the points on the blue interval at the bottom side of the polygon;
each point in the trapezoid is mapped to the point just below it, which is the point nearest to it on the polygon perimeter.
}
}
\end{figure}
Then, the perimeter can be partitioned like a 1-dimensional interval. 
\citet{HarPeled2021Partitioning} and \citet{Yagami2021Computing} present sketches of how this can be done, but again, they were not formally published so I do not claim any result for the cases in which $S$ is the family of connected polygons, or of simply-connected polygons.
}
%
%\footnote{
%\emph{Riemann's mapping theorem} implies that, 
%if $C$ is an open  simply-connected polygon, then there exists an analytic (infinitely differentiable) bijective mapping between $C$ and the open unit disc; in particular, such a mapping makes the cake convex. 
%While Riemann's mapping cannot be computed exactly, it can be approximated numerically using the Schwarz–-Christoffel mapping; see e.g. \citet{driscoll1996algorithm}.
%However, such a mapping does not necessarily make all existing pieces $Z_1,\ldots,Z_n$ convex.
%}

\paragraph{More than two dimensions.}
When the cake and pieces are three-dimensional (e.g. axes-parallel boxes), what is an upper bound on the number of blanks in a maximal expansion of a partial allocation?

%\related{4. Fatness.} A fat object is an object with a bounded length/width ratio, such as a square. Fatness makes sense in land division: if you are entitled to a 900 square meters of land, you will probably prefer to get them as a $30\times 30$ square or a $45\times 20$ rectangle rather than $9000\times 0.1$ sliver. A division problem with fatness requirement cannot be reduced to one-dimensional division. 
%There exist specialized division algorithms that support fatness constraints \cite{SegalHalevi2017Fair} and they can be used instead of the Even--Paz algorithm. However, I do not have an allocation-completion algorithm with fatness constraints. 

\section{Price-of-Fairness Bounds} \label{sec:pof-bounds}
This section uses the redivision theorems of previous sections to prove upper bounds on the price of partial\-/proportionality. 

\begin{theorem}
%	\label{thm:general-pof}
	For every $r\in[0,1]$, the utilitarian price of $r$-proportionality is at most $1/(1-r)$.
\end{theorem}
\begin{proof}
Let $\mathbf{Z}$ be a utilitarian-optimal allocation of a cake.
Apply Theorem \ref{thm:general+}
with $\mathbf{Z}$ as the original allocation.
The resulting division is $r$-proportional and satisfies $(1-r)$-ownership, so its utilitarian welfare is at least $1-r$ times the utilitarian welfare in $\mathbf{Z}$. 
\qed
\end{proof}

The proofs of Theorems \ref{thm:interval-pof}, \ref{thm:rectangle-pof} and \ref{thm:convex-pof} are similar; only the constants are different. 
The proof of Theorem \ref{thm:rectangle-pof} is presented below; to get the proofs of the other theorems, replace the constant ``3'' with ``2'' or ``4'' respectively.

\setcounter{theorem}{7}
\begin{theorem}
%	\label{thm:rectangle-pof}
	When the cake is a rectangle 
	and each piece must be a rectangle,
	for every $r\leq 1/3$:
	\begin{enumerate}
		\item The utilitarian-price of $r$\-/proportionality is $O(\sqrt{n})$;
		\item The Nash-price of $r$\-/proportionality is at most $(3 e)\cdot \exp {(1 / (4\pi e))} \approx 8.4$.\end{enumerate}
\end{theorem}
\begin{proof}
Part 1 is proved by Lemma \ref{thm:upop} below;
part 2 is proved by Lemma \ref{thm:npop} below.
\end{proof}

%The first part of Theorem \ref{thm:rectangle-pof}---that the \emph{utilitarian} price of partial\-/proportionality is in $O(\sqrt{n})$---uses the following lemma.
\begin{lemma}
\label{thm:upop}
Let $\mathbf{Z}$ be a utilitarian-optimal rectangular allocation of a cake $C$ among $n$ agents,
with $Z_1\sqcup\cdots\sqcup Z_n \subseteq C$.
Assume, without loss of generality, that the valuations are normalized such that \rev{$V_i(C)=n$} for all $i\in[n]$, so
the utilitarian welfare of $\mathbf{Z}$ is:
\begin{align*}
U:={1\over n}\sum_{i=1}^{n}{ V_i(Z_i) }
\end{align*}
Then, there exists a $(1/3)$\-/proportional rectangular allocation of $C$ to these same $n$ agents with utilitarian welfare $W$, such that $U/W\in O(n^{1/2})$.
\end{lemma}
\begin{proof}
Apply Theorem \ref{thm:rectangle+}
to the allocation $\mathbf{Z}$. The theorem ensures that the new division is $(1/3)$\-/proportional and satisfies democratic-ownership. The latter ensures that for every integer $d\range{1}{n-1}$, there is a set $S_d$ containing at least $n-d$ agents, such that the value of every $j\in S_d$ is \rev{larger than}
$
V_j(Z_j) / \lceil \frac{n}{d}\rceil
\geq 
V_j(Z_j) / ((n+1)/d)
%\geq 
%V_j(Z_j) / (2n/d)
=
d \cdot V_j(Z_j) / (n+1)$.

%$\max\left(\frac{d V_j(Z_j)}{2n},\frac{1}{3 n}\right)$. 

Renumber the agents in the following way. \rev{Choose} an agent from $S_{n-1}$ (which contains at least one agent) and number him/her $n-1$. \rev{Choose} an agent from $S_{n-2}$ (which contains at least one other agent) and number him/her $n-2$. Continue this way to number the agents by $d=n-1,\dots, 1$; renumber the remaining agent $0$. Now, the utilitarian welfare of the new division is lower-bounded by:

\begin{align*}
W
&> {1\over n}\sum_{d=0}^{n-1}\max\left(\frac{d\cdot  V_d(Z_d)}{(n+1)},\frac{1}{3 }\right )
%\\&\geq {1\over n}\cdot \frac{1}{3 n }\cdot\sum_{d=0}^{n-1}\max(d \cdot V_d(Z_d),~n) 
\\&\geq {1\over n}\cdot \frac{1}{3  }\cdot\sum_{d=0}^{n-1}\max(d \cdot V_d(Z_d)/n,~1)
\end{align*}
and the utilitarian welfare ratio is at most:

\begin{align*}
	\frac{U}{W} <  3 n\cdot\frac{\sum_{d=0}^{n-1}V_d(Z_d)/n}{\sum_{d=0}^{n-1}\max(d\cdot V_d(Z_d)/n,~1)}
\end{align*}
	\def\num{\textsc{Num}}
	\def\den{\textsc{Den}}
	\def\ratio{\frac{\num}{\den}}
Denote the right-hand side by \rev{$3n\cdot \ratio$}.  Let $a_d=V_d(Z_d)/n$, so that 
	
\begin{align*}
\num=\sum_{d=0}^{n-1}a_d
&&
\den=\sum_{d=0}^{n-1}\max(d\cdot a_d,~1).
\end{align*} 
Since $n$ is fixed, an upper bound on $U/W$ requires to find a sequence $a_0,\dots,a_{n-1}$ that maximizes $\ratio$ subject to $\forall d: 0\leq a_d\leq 1$.
	
%\footnote{This optimization problem is, interestingly, similar to the optimization problem used by \cite{Aumann2015Efficiency} to prove an upper bound on the price of envy-freeness in one dimension. However, the stories behind the two problems are quite different. }

\begin{observation}
In a maximizing sequence, $a_0=1$ and there is no $d>0$ such that $a_d<1/d$.
\end{observation}
\begin{proof}
If $a_0<1$, then setting $a_0$ to $1$ strictly increases $\num$ and does not change $\den$. 
Similarly, if $a_d < 1/d$ for some $d>0$,
then setting it to $1/d$ strictly increases $\num$ and does not change $\den$. 
\qed
\end{proof} 

\begin{observation}
A maximizing sequence must be weakly-decreasing (for all $d<d'$, $a_{d}\geq a_{d'}$).
\end{observation}
\begin{proof}
If there exists $d<d'$ such that $a_d<a_{d'}$, then swapping $a_d$ with $a_{d'}$ strictly decreases $\den$ and does not change $\num$.	\qed
\end{proof}

\begin{observation}
\rev{There exists at least one maximizing sequence in which} there is no $d>0$ such that $1/d<a_d<1$. 
\end{observation}
\begin{proof}
\footnote{The proof idea is due to Varun Dubey in http://math.stackexchange.com/q/1609071/29780 }
\rev{
Call an element $a_d$ ``bad'' if $1/d < a_d < 1$. 
Consider a maximizing sequence with the smallest number of bad elements. If the number of bad elements is $0$, then the proof is complete. 
Otherwise, pick one bad element $a_d$.
Let $\epsilon := \min(1 - a_d,~ a_d - 1/d)$. Since $a_d$ is bad, $\epsilon > 0$, and both $a_d+\epsilon$ and $a_d-\epsilon$ are in $[1/d,1]$. Replacing $a_d$ with $a_d + \epsilon$ or $a_d-\epsilon$ yields a new ratio, and it is at most the maximum ratio. In particular:
	\begin{itemize}
	\item Replacing $a_d$ with $a_d+\epsilon$ makes the ratio $\frac{\num + \epsilon}{\den + d\epsilon}$; this new ratio is at most $\ratio$, so $\epsilon\cdot \den \leq d \epsilon \cdot \num
	\implies \den \leq d \cdot \num$.
	\item Replacing $a_d$ with $a_d-\epsilon$ makes the ratio $\frac{\num - \epsilon}{\den - d\epsilon}$; that new ratio is at most $\ratio$, so $-\epsilon\cdot \den \leq -d \epsilon \cdot \num \implies \den \geq d\cdot \num$.
\end{itemize}
Moreover, at least one of these two inequalities is strict: if $\epsilon = 1-a_d$, then $a_d+\epsilon=1$, so replacing $a_d$ with $a_d+\epsilon$ yields a sequence with strictly fewer bad elements.
Similarly, if $\epsilon = a_d-1/d$, then replacing $a_d$ with $a_d-\epsilon$ yields a sequence with strictly fewer bad elements.
Since, by assumption, the maximizing sequence had a smallest number of bad elements, the new sequence must not be maximizing.  So either $\den \leq d\cdot \num$ and $\den > d\cdot\num$, or $\den < d\cdot \num$ and $\den \geq d\cdot\num$.
In both cases there is a contradiction. 
Therefore, there must exist a maximizing sequence with no bad elements.
}
\qed
\end{proof}
	
	Observations 1,2,3 imply that a maximizing sequence has a very specific format. It is characterized by an integer $l\in\{0,\dots,n-1\}$ such that, for all $d\leq l$, $a_d=1$ and for all $d\geq l+1$, $a_d= 1/d$. So:

\begin{align*}
\ratio 
&= \frac{
	\sum_{d=0}^{n-1}a_d
}{
\sum_{d=0}^{n-1}\max(d\cdot a_d,1)
}
\\
&= \frac{
(l+1) + (H_{n-1}-H_{l})
}{
\frac{1}{2}l(l+1) + (n-l-1)
} < \frac{
2(l + H_{n} + 1)
}{
l^2 - l + 2(n-1)
}
\end{align*}
where $H_n=\sum_{d=1}^n(1/d)$ is the $n$-th harmonic number.

The number $l$ is an integer, but the expression is upper-bounded by the maximum attained when $l$ is allowed to be real. By standard calculus \rev{(taking the derivative of the expression w.r.t. $l$, comparing the derivative to 0, and checking the second derivative)}, 
\iffalse  % Hidden explanation 
\begin{align*}
0 &= 
2 \cdot  (l^2 - l + 2n - 2)
-
2\cdot (l + H_n + 1)\cdot(2 l - 1)
\\
l^2 - l + 2n - 2
&=
(l + H_n + 1)\cdot(2 l - 1)
=
2 l^2 + 2 l H_n + 2 l
- l - H_n - 1
\\
0
&=
l^2 + 2 (H_n + 1) l  + 1 - H_n - 2n 
\\
l &= - (H_n + 1) \pm \sqrt{(H_n + 1)^2 - (1-H_n-2n)}
\\
&=
- (H_n + 1) \pm \sqrt{H_n^2 + 2 H_n + 1 - 1 + H_n + 2n}
\\
&=
- (H_n + 1) \pm \sqrt{H_n^2 + 3 H_n + 2n}
\end{align*}
\fi
the real value of $l$ which maximizes the above expression is 
\begin{align*}
l=\sqrt{H_n^2 + 3 H_n + 2n}-(H_{n}+1)
\end{align*}
which is in $\Theta(\sqrt{n})$. Substituting into the above inequality gives:

\begin{align*}
&& 
\ratio\leq \frac{\Theta(n^{1/2})}{\Theta(n)} \in \Theta(n^{-1/2})  
\\
\implies 
&&
\frac{U}{W}<3 n\cdot \ratio \in O(n^{1/2})
\end{align*}	
as claimed.
\qed\end{proof}

%The second part of Theorem \ref{thm:rectangle-pof}---that the \emph{Nash} price of partial\-/proportionality is at most $8.4$---uses the following lemma. 
\begin{lemma}
\label{thm:npop}
Let $\mathbf{Z}$ be a Nash-optimal rectangular allocation of a cake $C$ among $n$ agents,
with $Z_1\sqcup\cdots\sqcup Z_n \subseteq C$.
Assume the valuations are normalized such that \rev{$V_i(C)=n$} for all $i\in[n]$.
Let $U$ be the Nash welfare of $\mathbf{Z}$ (the geometric mean of the values), defined by
\begin{align*}
U^n=\prod_{i=1}^{n} V_i(Z_i).
\end{align*}
Then, there exists a $(1/3)$\-/proportional rectangular allocation of $C$ to these same $n$ agents with Nash welfare $W$, and $U/W \leq  8.4$.
\end{lemma}
\begin{proof}
	Apply Theorem \ref{thm:rectangle+} to the allocation $\mathbf{Z}$. 
	Renumber the agents as in Lemma \ref{thm:upop}. The Nash welfare of the new allocation, raised to the $n$-th power, can be bounded as:
	
\begin{align*}
W^{n} 
&> \prod_{d=0}^{n-1}\max\left(\frac{d\cdot V_d(Z_d)}{(n+1)},\frac{1}{3 }\right) 
\\
&\geq{\left(\frac{1}{3}\right)}^n
\cdot \prod_{d=0}^{n-1}\max(d\cdot V_d(Z_d)/n,~1)
\end{align*}	
	and the ratio of the new welfare to the previous welfare can be bounded as:
	
\begin{align*}
\frac{U^{n}}{W^{n}}
&<
3^{n}\cdot\frac{\prod_{d=0}^{n-1} V_d(Z_d)}{\prod_{d=0}^{n-1}\max(d V_d(Z_d)/n,~1)}
\\
&=
\frac{ (3 n) ^ {n} }{\prod_{d=0}^{n-1}\max(d,~n/V_d(Z_d))}
\end{align*}	
The numerator does not depend on the valuations, so the ratio is maximized when the denominator is minimized. This happens when each factor in the product is minimized. \rev{The $0$-th factor is at least $1$, since $V_d(Z_d)\leq n$. The $d$-th factor, for $d\geq 1$, is at least $d$.
Therefore,}
\begin{align*}
\frac{U^{n}}{W^{n}}
&<
\frac{(3 n)^{n}}{\prod_{d=1}^{n-1}d}
=
\frac{(3 n)^{n}}{(n-1)!}
\\
&=
\frac{n(3 n)^{n}}{n!}\approx\frac{n(3 n)^{n}}{\sqrt{2\pi n}(n/e)^{n}}=\sqrt\frac{n}{2\pi}\cdot(3 e)^{n}
\end{align*}
where $e$ is the base of the natural logarithm. Taking the $n$-th root gives 
\begin{align*}
\frac{U}{W} &< (3 e)\cdot \sqrt{n/2\pi}^{1/n}
= (3 e)\cdot \exp {\frac{\ln{n} - \ln{2\pi}}{2n}}.
\end{align*}
\rev{
Consider the expression $\frac{\ln{x} - \ln{2\pi}}{2x}$.
By taking its derivative, one finds that its global maximum over the positive real numbers is attained at $x = 2 \pi e$, and this maximum equals $1/(4\pi e)$. Substituting in the above expression gives
\begin{align*}
\frac{U}{W} < (3 e)\cdot \exp {(1 / (4\pi e))} \approx 8.4
\end{align*}
as claimed. \qed
}
\end{proof}

%The optimization problem in the proof of Theorem \ref{thm:upop} is, interestingly, similar to the optimization problem used in \cite{Aumann2015Efficiency} to prove an upper bound on the price of envy-freeness in one dimension. However, the stories behind the two problems are quite different.

\subsection{\textbf{Future Work}}
\label{sub:future3}
\rev{
Theorems \ref{thm:general-pof}--\ref{thm:convex-pof} 
invoke the question of whether the upper bounds proved in them are tight.
}

\paragraph{Utilitarian price of fairness.}
\rev{
There is a lower bound of $\Omega(\sqrt{n})$ on the utilitarian price of proportionality
for a cake with no geometric constraints
 \citep{Caragiannis2012Efficiency},
 as well as for an interval cake and interval pieces \citep{Aumann2015Efficiency}.
However, these lower bounds do not imply similar lower bounds for \emph{partial} proportionality. 
In fact, without geometric constraints, our Theorem \ref{thm:general-pof} shows that the price of partial-proportionality is $O(1)$.
Therefore, it is interesting to know which of the following two options is correct for a cake with geometric constraints (e.g. interval cake and interval pieces):
}
\begin{enumerate}
\item 
There is a lower bound of $\Omega(\sqrt{n})$ matching Theorems \ref{thm:interval-pof}--\ref{thm:convex-pof}, or ---
\item The actual price of partial\-/proportionality is $o(\sqrt{n})$, maybe even $O(1)$.
\end{enumerate}
\rev{
The latter option is particularly attractive, since it may lead to a feasible and practical compromise between fairness and social welfare.}

\paragraph{Nash price of fairness.}
It is known
that without geometric constraints, every Nash-optimal allocation is envy-free\rev{; see e.g. \citet{segalhalevi2018monotonicity}}.
Hence, such allocation is proportional, so the Nash price of $r$-proportionality is 1 for any $r\in[0,1]$. However, this is not true when the pieces must be connected (or rectangular, or convex). 
\rev{
Appendix \ref{sec:pof-lower} shows several lower bounds on the Nash price of $r$-proportionality with connectivity constraints. However, there is a substantial gap between these lower bounds and the upper bounds proved above.
}

\section{Related Work} \label{sec:related}

\subsection{\textbf{Dynamic Fair Division}} \label{sub:dynamic}
The cake redivision problem differs from several division problems studied recently.

\related{1. Dynamic resource allocation} \citep{Kash2014No,Friedman2015Dynamic,friedman2017controlled,huo2020incremental} is a common problem in cloud-computing environments. The server has several resources, such as memory and disk-space. Agents (processes) come and depart. The server has to allocate the resources fairly among agents. When new agents come, the server may have to take some resources from existing agents. The goal is to do the re-allocation with minimal disruption to existing agents. In these problems, the resources are \emph{homogeneous}, which means that the only thing that matters is what quantity of each resource is given to each agent. In contrast, the present paper considers a \emph{heterogeneous} cake, so the algorithms must decide which parts of the cake should be given to which agent.

\related{2. Population monotonicity} \citep{Thomson1983Fair,Moulin1990Fair,Moulin2004Fair,Thomson2011Fair,segal2018resource,segalhalevi2018monotonicity} is an axiom that describes a desired property of allocation rules. When new agents arrive and the same division rule is re-activated, the value of all old agents should be weakly smaller than in the initial allocation. This axiom represents the virtue of solidarity: if sacrifices have to be made to support an additional agent, then everybody should contribute. 

\rev{
Population monotonicity is related to a special case of the redivision model, in which the new agents have no share at all in the initial allocation. However, the redivision model differs in two important aspects. 
First, even in the special case of new agents with no initial share, there is no upper bound on the value allocated to the incumbent agents. On the contrary, the ownership requirements puts a \emph{lower} bound on their value in the new allocation. 
Second, the redivision model is more general, and relates to settings in which all agents already have a (possibly unfair) share in the initial allocation.
}
%In addition, while the common approach in economics is axiomatic and interested mainly in existence results, the approach here is constructive and aims to provide efficient redivision algorithms.

\related{3. Private endowment} in economics resource allocation problems means that each agent is endowed with an initial bundle of resources. Then, agents exchange resources using a market mechanism. The classic problem in economics involves homogeneous resources, but it has also been studied in the cake-cutting framework \citep{Berliant2004Foundation,Aziz2014Cake}. A basic requirement in these works is \emph{individual rationality}, which means that the final value allocated to each agent must be weakly \emph{larger} than the value of the initial endowment (note the contrast with the population monotonicity axiom). 
This requirement is not made in the redivision  problem as it is incompatible with fairness: since some agents may initially own no land, individual rationality would mean that they might not receive anything in the exchange.

\related{4. Online division} 
is a setting in which either the agents or the divided resources are not all available at the time of the division, but rather arrive at different times.
\citet{Walsh2011Online} studies the online division of a \emph{divisible} resource. The motivation is a birthday party in an office, in which some agents come or leave early while others come or leave late. It is required to give some cake to agents who come early while keeping a fair share to those who come late.
\citet{Aleksandrov2015Online} and \citet{benade2018make} study the online division of \emph{indivisible} items. The motivation
is the \emph{food-bank problem}, where a charity organization receives food donations and must decide on-line to whom each donation should be allocated.
In these papers, in contrast to the present paper, it is impossible to re-divide allocated resources, since they are consumed by their receivers. 

\related{5. Land reform} is the re-division of land among citizens. It has been attempted in numerous countries around the globe and in many periods throughout history. Some books on land reform are  \citet{Powelson1988Story,Bernstein2002Land,2006Promised,Lipton2009Land}.
The earliest recorded land-reform was done in ancient Egypt in the times of King Bakenranef, 8th century BC. The most recent land-reform act has been legislated in Scotland in 2016 AD. 
Balancing fairness and ownership rights is a major concern in such reforms  \citep{sellar2006great,Hoffman2013Why,Wightman2015Poor,MacInnes2015Land}.

\subsection{\textbf{Partial Proportionality}}
\label{sub:partial}
While proportionality is the most common criterion of fair cake-cutting, it is often relaxed to partial\-/proportionality in order to achieve additional goals:

\related{1. Speed:}
finding a proportional division takes $\Theta(n \log{n})$ queries, but finding an $r$-proportional division takes only $\Theta(n)$ queries, for some sufficiently small $r\leq 0.1$ 
\ifdefined\FULLVERSION \citep{Edmonds2006Balanced,Edmonds2008Confidently}.
\else
\citep[and others]{Edmonds2006Balanced}.
\fi

\related{2. Improving social welfare:}
proportional allocations may be socially inefficient; efficiency can be improved by decreasing the value-guarantee per agent \citep{Zivan2011Can,Arzi2012Cake}.

\related{3. Minimum-size constraint:}
In some 1-dimensional settings, each agent may get several intervals but the length of each interval should be above a threshold. It is impossible to guarantee an $r$-proportional allocation for any $r>0$, but additive approximations exist \citep{Caragiannis2011Towards}.

\related{4. Geometric constraints:} For example, when the cake is square and the pieces must be square, it is impossible to guarantee an $r$-proportional allocation for any $r> 1/2$, but there is an algorithm that guarantees a $1/4$-proportional allocation
\citep{SegalHalevi2017Fair,segal2020envy}.
When the cake is a \rev{connected graph, and the pieces must be connected}, there is an algorithm that guarantees each agent at least $1/(2n-1)$ of the total value, and it is impossible to guarantee more than that \citep{BeiSu21}.

\subsection{\textbf{Democratic fairness}}
While most works on fair division aim to guarantee unanimous fairness, this is not always compatible with other requirements. Hence, some works explore the possibility of guaranteeing fairness to a subset of the agents, which should---ideally---be as large as possible. Some examples are:

\related{1. }
Envy-free allocation of multiple cakes (where each agent should receive a piece in each cake) to $(n+1)/2$ agents
\citep{nyman2020fair}.

\related{2. } 
Maximin-share fair allocation of indivisible objects to $n-1$ or $2n/3$ agents 
\citep{searns2020rethinking,hosseini2020mms}.

\related{3. }
Stable matching rules that guarantee resource-monotonicity to $n/2$ agents
\citep{ortega2018social}.

\related{4. }
Pricing rules that are envy-free to a pre-selected subset of the buyers \citep{bilo2018pricing}.

\subsection{\textbf{Geometric Cake Models}}  \label{sub:nicely}
The most prominent cake-model is a one-dimensional interval, in which case the pieces are often required to be contiguous sub-intervals. Some exceptions are:

\related{1. }The cake is a 1-dimensional circle (``pie'') and the pieces are contiguous arcs
\citep{Thomson2007Children,Brams2008Proportional,Barbanel2009Cutting,ElkindSeSu21}.
  
\related{2. }The cake is the union of edges of a connected graph, and the pieces are contiguous sub-graphs
\citep{BeiSu21}.

\related{3. }The cake is a 2-dimensional territory that lies among several countries. Each country should receive a piece adjacent to its border \citep{Hill1983Determining,Beck1987Constructing}.

\related{4. }The cake is 2-dimensional and the pieces are rectangles determined by the agents \citep{Iyer2009Procedure}.

\related{5. }The cake is 2-dimensional and the pieces must be squares or fat polygons \citep{SegalHalevi2017Fair,segal2020envy}.

\related{6. }The cake is 2-dimensional; the geometric constraints are connectivity or convexity \citep{Devulapalli2014Geometric}.

\related{7. }The cake is multi-dimensional and the pieces are simplices or polytopes
\citep{Berliant1992Fair,Ichiishi1999Equitable,DallAglio2009Disputed}.

\rev{
Very recently, geometric fair division problems have been studied based on real two-dimensional land-value data \citep{aleskerov2019allocation,shtechman2020fair}.
}

Many natural 2-dimensional settings have not been studied yet. For example, the setting studied in Section \ref{sec:rectangle}, 
where the cake is a rectilinear polygon and the pieces should be rectangles, has not been studied. 
As shown by Figure \ref{fig:disposal}, there is a qualitative (not only quantitative) difference between 2-D and 1-D division. 
2-D division introduces interesting paradoxes, that might be missed by the habit of assuming a one-dimensional cake. 

It is important to distinguish geometric cake-cutting from the \emph{geometric knapsack} problem \citep{arkin1993geometric,adamaszek2015quasi}.
In the latter there is a \emph{single} value-function that should be optimized. In cake-cutting, there are $n$ agents with \emph{different} value-functions, and the goal is to guarantee each agent a value higher than some threshold.

\subsection{\textbf{Price of Fairness}}
The price-of-fairness in cake-cutting has been studied in two settings:
\begin{itemize}
	\item The cake is a \emph{one-dimensional interval} and the pieces must be \emph{intervals} \cite{Aumann2015Efficiency}. The utilitarian-price-of-proportionality in this case is $\Theta(\sqrt{n})$.
	\item The cake is \emph{arbitrary} and the pieces may be \emph{arbitrary} \cite{Caragiannis2012Efficiency}. The utilitarian-price-of-proportionality in this case is $\Theta(\sqrt{n})$ too.
\end{itemize}
Both papers study the price of other fairness criteria such as envy-freeness and equitability, but do not study the price in Nash-welfare, and do not handle two-dimensional geometric constraints such as rectangularity or convexity.
	%\footnote{Note that, since both maxima in (*) are bounded by the geometric constraints, the price-of-fairness may be either higher or lower than with arbitrary pieces.}

\rev{
The price of fairness was also studied in the context of allocating homogeneous resources \citep{Bertsimas2011Price,Bertsimas2012EfficiencyFairness}, fair subset sum \citep{nicosia2017price}, kidney exchange \citep{Dickerson2014Price},
connected chore cutting \citep{heydrich2015dividing},
indivisible object allocation \citep{Caragiannis2012Efficiency,kurz2016price,bei2019pricefair,suksompong2019fairly,barman2020optimal},
budget division \citep{michorzewski2020price,tang2020price} 
and machine scheduling \citep{agnetis2019price,zhang2020price}.
}
 
A related notion---\emph{the price of connectivity}---was studied both for cake-cutting \citep{ram2018price} and for indivisible objects  \citep{bei2019pricecon}.

The Nash-price of fairness is related to results about approximating the maximum Nash welfare with indivisible goods. The approximation factors range from $2.89$ \citep{cole2015approximating} to $e$ \citep{anari2017nash} to $2$ \citep{cole2017convex,mcglaughlin2020improving,caragiannis2019envy} to $1.45$ \citep{barman2018finding}.

Several authors study the algorithmic problem of finding a welfare\-/maximizing cake-allocation  in various settings:

\related{1. } The cake is an interval and the pieces must be connected \citep{Aumann2013Computing};

\related{2. } The cake is an interval and the pieces must be connected, and additionally, the division must be proportional \citep{Bei2012Optimal};

\related{3. } The cake and pieces are arbitrary, and the division must be envy-free \citep{Cohler2011Optimal}.

\related{4. } The cake and pieces are arbitrary, and the division must be equitable \citep{Brams2012Maxsum}.

\section{More Future Work} \label{sec:future}
Besides the open questions mentioned in subsections \ref{sub:future1}, \ref{sub:future2} and \ref{sub:future3},
it may be interesting to study the redivision problem with other requirements besides proportionality.

\paragraph{Envy-freeness.} 
Envy-freeness means that each agent values their piece at least as much as each of the other pieces. Similarly, $r$-envy-freeness means that each agent values their piece as at least $r$ times the value of each of the other pieces. For what pairs $r,w$ is $r$-envy-freeness compatible with $w$-ownership? With democratic-ownership?%

\rev{
One issue with envy-freeness is that the redivision problem is inherently asymmetric: agents whose initial piece is valuable are entitled to a higher final value than agents whose initial piece is empty.
A potentially useful notion here is \emph{justified envy}, which has been recently studied in the literature on two-sided matching \citep{abdulkadirouglu2020efficiency}.
In two-sided matching (for example, between doctors and hospitals), ``justified envy'' means that doctor $d_1$, who is matched to hospital $h_1$, envies doctor $d_2$, who is matched to hospital $h_2$, and at the same time, $h_2$ prefers $d_1$ to $d_2$.
Analogously, one can defined ``justified envy'' in our setting as some agent $i$ envying another agent $j$ in the final allocation, while $i$'s initial piece was more valuable than $j$'s.
}

\paragraph{Pareto-efficiency.} From an existential point of view, Pareto-efficiency does not add much difficulty. Both $r$-proportionality and $w$-ownership are preserved by Pareto-improvements. Therefore, if there exists a division satisfying $r$-proportionality and $w$-ownership (or democratic-ownership), then there also exists a Pareto-optimal division satisfying these properties. However, it may not be easy to find such a division algorithmically.

\newpage

\section*{Acknowledgements}
I am currently supported by Israel Science Foundation grant 712/20.

I started this work during my Ph.D. studies in Bar-Ilan university guided by Yonatan Aumann and Avinatan Hassidim \citep{SegalHalevi2017Phd}.
I am grateful to Ioannis Caragiannis 
for introducing me to the Nash welfare and Nash price of fairness in the COST Summer School on Fair Division in Grenoble, 7/2015 (FairDiv-15).

I benefited a lot from discussions with 
participants in BIU game-theory seminar, Glasgow university micro-economics seminar and Corvinus university game-theory seminar, 
particularly: Reuven Cohen, Herve Moulin, Yehuda Levy, Panagiotis Kanellopoulos, Laszlo Csato and Aris Filos-Ratsikas.

I also received a lot of mathematical help from Chris Culter, Varun Dubey, Alex Ravsky, Swami Sarvattomananda, 
Zhen Lin,
% https://math.stackexchange.com/q/400417/29780
Sariel Har-Peled, Saeed Amiri, David Eppstein,
% https://cstheory.stackexchange.com/q/48462/9453
Inuyasha Yagami,
% https://cs.stackexchange.com/a/141012/1342
David K, \citet{Justpassingby2016Connectivity},
% https://math.stackexchange.com/q/1637220/29780
j\_random\_hacker,
% https://cs.stackexchange.com/q/80574/1342
and WoolierThanThou.
% https://math.stackexchange.com/q/4037900/29780

I am grateful to 
the anonymous reviewers
of AAMAS 2016, EC 2016, SODA 2017, ESA 2017, IJCAI 2018 and the AAMAS journal for their very helpful comments.

Above all, to Galya Segal-Halevi, thanks for all the cakes!

\newpage
\appendix
\section{\rev{Lower bounds on price of $r$-proportionality}}
\label{sec:pof-lower}

This appendix presents several lower bounds on the  price of $r$-proportionality with geometric constraints.
As a warm-up, the following simple lower bound is proved (the bound for utilitarian price of proportionality was already proved by \citet{Aumann2015Efficiency}):
\begin{theorem}
With $n=2$ agents, 
when the cake is an interval and the pieces must be intervals,
the utilitarian price of proportionality is $3/2$
and
the Nash-price of proportionality is $\sqrt{2}$.
\end{theorem}
\begin{proof}
For the lower bound,%
\footnote{It extends an example in Section 5 of \citet{segal2018resource}.}
consider an interval cake with four homogeneous regions, where the values of each of two agents for each of the four regions is given by the table below (where $\epsilon>0$ is an infinitesimally small positive constant):

\begin{center}
\begin{tabular}{|c|c|c|c|c|}
\hline
George's value: & $1-\epsilon$ & $\epsilon$ & $\epsilon$ & $1-\epsilon$ \\
\hline
Alice's value:  & $\epsilon$ & $1-\epsilon$ & $1-\epsilon$ & $\epsilon$ \\
\hline
\end{tabular}
\end{center}

Due to symmetry, the only connected proportional allocation divides the cake exactly in the middle and gives each agent exactly $1$, so both the utilitarian welfare and the Nash welfare are $1$. However, giving the leftmost region to George and the rightmost three regions to Alice gives Alice a value of almost $2$ and George a value of almost $1$, so the 
utilitarian welfare is almost $3/2$
and 
Nash welfare is almost $\sqrt{2}$.
When $\epsilon\to 0$, 
the utilitarian price of proportionality approaches $3/2$ and
the Nash price of proportionality approaches $\sqrt{2}$.

For the matching upper bound, note that,
in any proportional allocation,
the utilitarian welfare and
the Nash welfare are at least $1$ (it is attained when all agents get exactly their proportional share).
On the other hand, 
in any non-proportional allocation,
the utilitarian welfare is less than 
$n-1+\frac{1}{n}$,
and the Nash welfare is less than $(n^{n-1})^{1/n}$ (since one agent gets a value of less than $1$, and the other agents get a value of at most $n$). Therefore, 
the utilitarian price of proportionality is 
at most $n-1+\frac{1}{n}=3/2$
and the Nash price of proportionality is at most $n^{1-1/n}=\sqrt{2}$. \qed
\end{proof}

Below, this lower bound is extended in two ways: two agents with $r$-proportionality and 2-dimensional cakes,
and $n$ agents with $1$-proportionality.

\begin{theorem}
With $n=2$ agents, 
interval cake and interval pieces,
or rectangular cake and rectangular pieces,
for any $r \in [0, 1]$,
the utilitarian price of $r$-proportionality is $1+r/2$
and
the Nash price of $r$-proportionality is $\max(1,\sqrt{2 r})$.
\end{theorem}

\begin{proof}
For the lower bound, 
consider an interval cake with six homogeneous interval regions, 
or a rectangular cake of dimensions $6\times 6$ with six homogeneous rectangular regions of dimensions $1\times 6$ each, where the values of each of two agents for each region are given below (where $\epsilon>0$ is an infinitesimally small positive constant):

\begin{center}
\begin{tabular}{|c|c|c|c|c|c|c|}
\hline
George's value: 
& $r-\epsilon$ & $\epsilon$ & $1-r$ & $1-r$ & $\epsilon$ & $r-\epsilon$
\\
\hline
Alice's value:  
& $\epsilon$ & $r-\epsilon$ & $1-r$ & $1-r$ & $r-\epsilon$ & $\epsilon$
\\
\hline
\end{tabular}
\end{center}

\begin{observation}
In any $r$-proportional allocation, 
both 
the 
utilitarian and the Nash welfare are at most $1$.
\end{observation}
\begin{proof}
Suppose first that the cake is an interval.
In any $r$-proportional allocation, each agent must get a value of at least $r$. Hence, one agent must get the two leftmost regions and the other agent must get the two rightmost regions.
The two central regions can be divided arbitrarily;
regardless of how they are divided, the utilitarian welfare is $1$.
To compute the maximum Nash welfare, suppose George receives the two leftmost regions and a fraction $x$ of the two central regions; so his value is $r + 2x(1-r)$, while Alice's value is $r + 2(1-x)(1-r)$. 
By taking the derivative, one can find out that the product is maximized when $x=1/2$, where the value of each agent is exactly $1$. Hence, the maximum Nash welfare in an $r$-proportional allocation is $1$ too.
\iffalse
The product is 
$(r+xw)(r+(1-x)w)$
Taking derivative:
$
w(r+(1-x)w) - w(r+xw)
=
w(r+(1-x)w-r-xw)
=
w^2(1-2x)
x=1/2
$
\fi

If the cake is rectangular, then there are two ways to divide it into two rectangles: vertically or horizontally. If it is divided by a vertical cut, then the situation is exactly as in the case of an interval.
If it is divided by a horizontal cut, then --- regardless of the cut location --- the sum of the agents' values is $2$, so the utilitarian welfare is $1$ and the maximum Nash welfare is attained when both values are $1$. In both cases, the maximum utilitarian and Nash welfare in an $r$-proportional allocation is $1$.
\end{proof}

\begin{observation}
\label{obs:sqrt2r}
There is an allocation
in which, when $\epsilon\to 0$,
the utilitarian welfare approaches $1+r/2$
and the Nash welfare approaches $\sqrt{2 r}$.
\end{observation}
\begin{proof}
Cut the cake at the right of the leftmost region (if the cake is rectangular, use a vertical cut).
Give George the leftmost region and Alice the other regions.
George's value is $r-\epsilon$ while Alice's value is $2-\epsilon$, so when $\epsilon\to 0$,
the utilitarian welfare approaches $1+r/2$
and
the Nash welfare approaches $\sqrt{2 r}$.
\end{proof}

The above two observations imply 
a lower bound of $1+r/2$
on the utilitarian price of $r$\-/proportionality
and a lower bound of 
$\max(1,\sqrt{2 r})$
on the Nash price of $r$-proportionality.%
\footnote{The max in the latter expression comes from the fact that, when $\sqrt{2r}<1$, the maximum Nash welfare is not attained by the allocation of Observation \ref{obs:sqrt2r}, but rather by a proportional allocation.}

For the matching upper bound, note that in any proportional allocation (which is, in particular, $r$-proportional), 
both the utilitarian and the Nash welfare are at least $1$, while in any non-$r$-proportional allocation, 
the utilitarian welfare is less than $1+r/2$
the Nash welfare is less than $\sqrt{2 r}$ (since one agent gets less than $r$ and the other agent gets at most $2$). 
\qed
\end{proof}

\begin{theorem}
\label{thm:pof-nash-n}
When the cake is an interval and the pieces must be intervals,
the Nash price of proportionality when there are $n$ agents is at least $2^{1-1/n}$. 
\end{theorem}
\begin{proof}
Consider a piecewise-homogeneous cake consisting of $2 n$ regions.
The agents are partitioned into two groups: odd-indexed agents and even-indexed agents. The agents in each group have the same valuation. 
The odd-indexed agents value the regions at $1-\epsilon, \epsilon, \epsilon, 1-\epsilon, \ldots$, 
while the even-indexed agents value the regions at $\epsilon, 1-\epsilon, 1-\epsilon, \epsilon, \ldots$, 
An example is shown in the table below for $n=5$, where $1^-$ is a shorthand for $1-\epsilon$:

\begin{center}
\begin{tabular}{|c|c|c|c|c|c|c|c|c|c|c|}
\hline
Agents $1,3,5$:
& $1^-$ & $\epsilon$ & $\epsilon$ & $1^-$ 
& $1^-$ & $\epsilon$ & $\epsilon$ & $1^-$ 
& $1^-$ & $\epsilon$ 
\\
\hline
Agents $2,4$:
& $\epsilon$ & $1^-$ & $1^-$ & $\epsilon$ 
& $\epsilon$ & $1^-$ & $1^-$ & $\epsilon$ 
& $\epsilon$ & $1^-$
\\\hline
\end{tabular}
\end{center}

The total cake value for all agents is $n$, so in a proportional allocation, each agent should get a value of at least $1$. The following two observations imply the theorem.

\begin{observation}
In a proportional allocation, the value of every agent is exactly $1$; hence the Nash welfare is $1$.
\end{observation}
\begin{proof}
In a proportional allocation, the agent who receives the leftmost piece must receive at least two adjacent regions in order to have a value of at least $1$.

The two leftmost agents must receive together at least four adjacent regions. This is because the value measure arrives at $2$ only at the end of the fourth region, and the leftmost agent consumes at least two regions which are worth at least $1$, so to ensure the second-leftmost agent a value of at least $1$, their combined consumption must consist of at least the four leftmost regions.

Proceeding this way, it is possible to prove by induction that, for every integer $\ell\geq 1$, the $\ell$ leftmost agents consume together at least $2\ell$ regions. By a symmetric argument, the same is true for the $\ell$ rightmost agents. But this implies that, in a proportional allocation, each agent must consume exactly $2$ regions. The value of every two consecutive regions when starting from the left (or from the right) is exactly $1$. 
\end{proof}

\begin{observation}
There is an allocation in which the Nash welfare approaches $2^{1-1/n}$ as $\epsilon\to 0$.
\end{observation}
\begin{proof}
Give the leftmost region to agent $1$.
Then, give each of agents $2,\ldots,n-1$ two consecutive regions.
Give agent $2$ the last three consecutive regions.

The value of agent $1$ is $1-\epsilon$; 
the value of each of $2,\ldots,n-1$ is $2-2\epsilon$;
and the value of $n$ is $2-\epsilon$.
When $\epsilon\to0$, the Nash welfare approaches $(2^{n-1})^{1/n} = 2^{1-1/n}$.
\end{proof}

The above two observations imply that the Nash price of proportionality approaches $2^{1-1/n}$ when $\epsilon\to 0$.%
\footnote{
The same example implies that the utilitarian price of proportionality approaches $2-1/n$. But this is subsumed by the $\Omega(\sqrt{n})$ lower bound of \citet{Aumann2015Efficiency}.
}
\qed
\end{proof}

So far, I could not extend Theorem \ref{thm:pof-nash-n} to rectangular cakes: the main difficulty is that there are many possible ways to cut a rectangle into $n$ rectangles, so it is hard to reason about what the possible $r$-proportional allocations can be. 
Similarly, I could not extend the theorem to $r$-proportionality: giving even a single agent a value of $r$ apparently allows a lot of freedom in allocating to the other agents, so again, it is hard to reason about what an $r$-proportional allocations can be.
Thus, the exact utilitarian and Nash price of $r$-proportionality for all $n\geq 3$ and $r\in(0,1)$ remains open.

%\newpage
%\bibliographystyle{named}
%\bibliography{../../erelsegal-halevi}

\end{document}